\documentclass[12pt]{article}

\usepackage{amsmath,amsthm,amsfonts,amssymb,amscd,cite}
\usepackage{graphicx}

\allowdisplaybreaks[4]

\newtheorem {lemma}{Lemma}
\newtheorem {theorem} {Theorem}

\newtheorem {corollary}{Corollary}

\begin{document}

\title{Link residual closeness of graphs with fixed parameters}

\author{Leyou Xu\footnote{E-mail: leyouxu@m.scnu.edu.cn}, Chengli Li\footnote{E-mail: lichengli@m.scnu.edu.cn}, Bo Zhou\footnote{
E-mail: zhoubo@scnu.edu.cn}\\
School of  Mathematical Sciences, South China Normal University,\\
Guangzhou 510631, P.R. China}

\date{}
\maketitle

\begin{abstract}
Link residual closeness is a newly proposed measure for network vulnerability. In this
model, vertices are perfectly reliable and the links fail independently of each other.
It measures  the vulnerability even when the removal of links does not disconnect the graph.
In this paper, we characterize those graphs that maximize the link residual closeness over the connected graphs with fixed order and one parameters such as connectivity, edge connectivity, bipartiteness, independence number, matching number, chromatic number,  number of vertices and number of cut edges.
\end{abstract}

\section{Introduction}

The vulnerability of a network is  the measurement of the global strength of its underlying graph where the vertices represent the processing elements of the system and the edges (or links) connect pair of vertices
that mutually interact exchanging information \cite{HK,Ja}.
In understanding  of computer networks, how to protect a
network from vulnerability or to improve network robustness remains an overarching concern.

It is highly desirable to identify a class of easily
computed measures that characterize network vulnerability.
There are lots of different measures for network vulnerability such as  connectivity, toughness, scattering number, binding number, and their link counterparts, see, e.g. \cite{Fr2,Chv,Ju,Woo}.
These measures may be used if network failure (by the removal of vertices or links) means that the underlying graph has become disconnected or trivial.  To measure the vulnerability
even when the removal of vertices/links does not disconnect the graph,
Dangalchev \cite{Dan} proposed a new type vulnerability measure that is called
residual closeness (including vertex and link versions).
It was argued in \cite{Dan} that vertex (link, respectively) residual closeness is the most appropriate approach for modeling the robustness of network topologies in the face of possible vertex (link, respectively) destruction.
The vertex version has received a lot of attention, see, e.g. \cite{AO1,AO2,AO3,Dan2,Dan3,OA1} for the computational aspects and \cite{CZ,WZ,ZLG} for the extremal properties.
However, the link version received less attention.
Berberler and Yi\v{g}it established formulae for the link residual closeness of  path-type graphs such as regular caterpillars in \cite{BY}, wheel type graphs in \cite{YB2}, and
composite  graphs such as graph unions and graph joins in \cite{YB}.  Some preliminary extremal properties of the link residual closeness were given in \cite{ZLG}. For example, the trees with minimum and maximum link residual closeness were determined there.

Let $G$ be a graph with vertex set $V(G)$ and edge $E(G)$.
For a graph $G$ with $u,v\in V(G)$, the distance between $u$ and $v$ in $G$, denoted by $d_G(u,v)$, is the length of a shortest path connecting them in $G$, and let $d_G(u,v)=\infty$ if there is no path connecting $u$ and $v$ in $G$. In particular, $d_G(u,u)=0$ for any $u\in V(G)$.
For a vertex $u$ of a graph $G$, the closeness of $u$ in $G$ is defined as
$C_G(u)=\sum_{v\in V(G)\setminus\{u\}}2^{-d_G(u,v)}$.
The   closeness of a graph $G$ is defined as
\[
C(G)=\sum_{u\in V(G)}C_G(u)=\sum_{u\in V(G)}\sum_{v\in V(G)\setminus\{u\}}2^{-d_G(u,v)}.
\]
The link residual closeness of a nonempty graph $G$ is defined as \cite{Dan}
\[
R^L(G)=\min\{C(G-e): e\in E(G)\}.
\]
If $G$ is empty, then we set $R^L(G)=0$. It is evident that $R^L(G)=0$ if and only if $|E(G)|=0,1$.
For completeness, we mention that vertex residual closeness of a nontrivial graph $G$ is defined as \cite{Dan}
\[
R(G)=\min\{C(G-v): v\in V(G)\}.
\]

To have a fuller understanding of the relationship between the link residual closeness and the structural properties of the graphs, we consider the extremal problems to maximize the link residual closeness in some families of connected graphs.  In this paper, we identify those graphs that maximize the link residual closeness in the families of connected graphs of fixed order and one of the parameters such as connectivity, edge connectivity, bipartiteness, independence number, matching number, chromatic number, number of cut vertices and number of cut edges.

\section{Preliminaries}

For a vertex $u$ of a graph $G$, the neighborhood of $u$ in $G$ is the set  $N_G(u)=\{v: uv\in E(G)\}$ and the degree of $u$ in $G$ is $|N_G(u)|$, denoted by $\delta_G(u)$.

For a nonempty proper subset $V_1$ of vertices of a graph $G$,  $G-V_1$ denotes the subgraph of $G$ obtained by deleting all vertices in $V_1$ (and the incident edges) from $G$, and in particular, if $V_1=\{u\}$, then we write $G-u$ for $G-\{u\}$.
For a subset $E_1$ of edges of a graph $G$,   $G-E_1$ denotes the subgraph obtained from $G$ by deleting all edges in $E_1$, and in particular, if $E_1=\{e\}$, then we write $G-e$ for $G-\{e\}$. The complement $\overline{G}$ of a graph $G$ is the graph with vertex set $V(G)$ so that two vertices are adjacent in $\overline{G}$ if and only they are not adjacent in $G$.
 For a graph $G$ with $E_2\subseteq E(\overline{G})$,  $G+E_2$ denotes the graph obtained from $G$ by adding all edges in $E_2$, and we write $G+e$ for $G+\{e\}$ when $E_2=\{e\}$.

For a nonempty  subset $V_1$ of vertices of a graph $G$, $G[V_1]$ denotes the subgraph of $G$ induced by $V_1$. 
For vertex disjoint graphs $G_1$ and $G_2$, the union of $G_1$ and $G_2$, denoted by $G_1\cup G_2$, is the graph with vertex set $V(G_1) \cup V(G_2)$ and  edge set $E(G_1)\cup E(G_2)$. The join of $G_1$ and $G_2$, denoted by $G_1\vee G_2$, is the graph $(G_1\cup G_2)+\{v_1v_2: v_1\in V(G_1), v_2\in V(G_2)\}$, i.e.,
$G_1\vee G_2= \overline{\overline{G_1} \cup \overline{G_2}}$.

Let $K_{n_1,\dots,n_k}$ be the complete $k$-partite graph with coloring class sizes (or partite sizes) $n_1, \dots, n_k$. Let $S_n=K_{1, n-1}$. A complete $k$-partite graph $K_{n_1,\dots,n_k}$ is said to be balanced if $|n_i-n_j|\le 1$ for any $i,j$ with  $1\le i<j\le k$. Let $K_n$ be the $n$-vertex complete graph. Let $P_n$ be the $n$-vertex path.

%


\begin{lemma}  \label{AddE}
Let $G$ be a nonempty graph  in which two vertices $u$ and $v$ are not adjacent.   Then $R^L(G) <  R^L(G+uv)$.
\end{lemma}

\begin{proof} Let $H=G+uv$.
Assume that $R(H)=C(H-e)$ with $e\in E(H)$.
If $e=uv$, then $R^L(H)=C(H-e)=C(G)>C(G-f)\ge R^L(G)$ for any $f\in E(G)$.
If $e\ne uv$, then $R^L(H)=C(H-e)>C(H-e-uv)=C(H-uv-e)=C(G-e)\ge R^L(G)$.
\end{proof}

If $G=K_{n_0} \vee (K_{n_1}\cup\dots\cup K_{n_t})$, then we call $K_{n_i}$ the $i$-th inner copy of $G$, where $i=1,\dots, t$.

\begin{lemma}\label{kst}
For integers $t\ge 2$, $n_0\ge 1$ and $1\le n_1\le \dots\le n_t$,
let $G=K_{n_0} \vee (K_{n_1}\cup\dots\cup K_{n_t})$ and $n=\sum_{i=0}^{t}n_i$.
If $n_0=1$, then
\[
R^L(G)=\begin{cases}
\frac{1}{4}\sum_{i=2}^{t} n_i^2+\frac{1}{4}n^2-\frac{1}{2}n &\mbox{ if }n_1=1,\\
\frac{1}{4}\sum_{i=1}^t n_i^2+\frac{1}{4}n^2-\frac{1}{4}n+\frac{1}{4}n_1-\frac{1}{2} &\mbox{ if }n_1\ge 2.
\end{cases}
\]
Otherwise,
\[
R^L(G)=\frac{1}{4}\sum_{i=1}^{t} n_i^2+\frac{1}{4}n^2+\frac{n_0-1}{2}n-\frac{1}{4}n_0^2-\frac{1}{2}.
\]
\end{lemma}

\begin{proof}
It can be easily seen that
\begin{align*}
C(G)&=n_0\sum_{i=1}^{t}n_i+\sum_{i=0}^{t}{n_i\choose 2}+\frac{1}{2}\sum_{1\le i<j\le t} n_in_j\\
&=n_0(n-n_0)+\frac{1}{2}(n_0^2-n_0)\\
&\quad +\frac{1}{4}\left(\sum_{i=1}^t n_i^2+(n-n_0)^2-2(n-n_0)\right)\\
&=\frac{1}{4}\sum_{i=1}^t n_i^2+\frac{1}{4}n^2+\frac{n_0-1}{2}n-\frac{1}{4}n_0^2.
\end{align*}	
	
\noindent
{\bf Case 1.} $n_0=1$.

Denote by $u$  the unique vertex with degree $n-1$ in  $G$.
Let $v\in V(G)\setminus\{u\}$.

If the degree of $v$ is one, then
\begin{align*}
C(G-uv) &=C(K_1\vee (K_{n_2}\cup\dots\cup K_{n_t}))\\
&=C(G)-1-\frac{1}{2}\sum_{i=2}^tn_i,
\end{align*}
which is minimum for if $v$ is in the $1$-st inner copy of $G$ because $n_1\le \dots\le n_t$.

If $v$ is of degree at least two, then $v$ is in the $\ell$-th inner copy of $G$ such that $n_{\ell}\ge 2$ for some $\ell=1,\dots,k$,
$d_{G-uv}(u,v)=2$, $d_{G-uv}(v,z)=3$ if $z$ is any vertex in the $i$-th inner copy with $i=1,\dots, t$ and  $i\ne \ell$, and
as we pass from $G$ to $G-uv$, the distance between any other vertex pair  remains unchanged, so
\begin{align*}
C(G-uv) & =C(G)-1+\frac{1}{2}-\frac{1}{2}\sum_{i=1,i\ne \ell}^{t}n_i+\frac{1}{4}\sum_{i=1,i\ne \ell}^{t}n_i\\
&=C(G)-\frac{1}{2}-\frac{1}{4}\sum_{i=1,i\ne \ell}^{t}n_i,
\end{align*}
which is minimum for $\ell=1,\dots, t$ if $\ell=1$ because $n_1\le \dots\le n_t$.

Now, let $wz$ be any edge of $G$ with $w,z\ne u$.  Then both the degrees of $w$ and $z$ are at least two.  Evidently, $d_G(w,z)=1$, $d_{G-wz}(w,z)=2$, and
as we pass from $G$ to $G-wz$, the distance between any other vertex pair  remains unchanged, so
\[
C(G-wz)=C(G)-1+\frac{1}{2}>C(G-uv)
\]
whether the degree of $v$ is one or at least two.

If $n_1=1$, then $C(G)-1-\frac{1}{2}\sum_{i=2}^tn_i<C(G)-\frac{1}{2}-\frac{1}{4}\sum_{i=1,i\ne \ell }^{t}n_i$ for any $\ell=1,\dots, t$,
so $R^L(G)=C(G-uv)$ with $v$ in the $1$-st inner copy of $G$. That is,
\begin{align*}
R^L(G)&=C(G)-1-\frac{1}{2}\sum_{i=2}^tn_i\\
&=\frac{1}{4}\sum_{i=2}^{t} n_i^2+\frac{1}{4}n^2-\frac{1}{2}n.
\end{align*}

If $n_1\ge 2$, then
\begin{align*}
R^L(G)&=C(G)-\frac{1}{2}-\frac{1}{4}\sum_{i=2}^tn_i\\
&=\frac{1}{4}\sum_{i=1}^t n_i^2+\frac{1}{4}n^2-\frac{3}{4}-\frac{1}{4}(n-1-n_1)\\
&=\frac{1}{4}\sum_{i=1}^t n_i^2+\frac{1}{4}n^2-\frac{1}{4}n+\frac{1}{4}n_1-\frac{1}{2}.
\end{align*}

\noindent
{\bf Case 2.} $n_0\ge 2$.

For any $wz\in E(G)$, as we pass from $G$ to $G-wz$, the distance between $w$ and $z$ is changed from $1$ to $2$, and the distance between any other vertex pair  remains unchanged, so $C(G-wz)=C(G)-\frac{1}{2}$.
It follows that
\[
R^L(G)= C(G)-\frac{1}{2}=\frac{1}{4}\sum_{i=1}^t n_i^2+\frac{1}{4}n^2+\frac{n_0-1}{2}n-\frac{1}{4}n_0^2-\frac{1}{2},
\]
as desired.
\end{proof}

A path $u_0\dots u_{\ell}$ in a graph $G$ is a pendant path of length $\ell$ of $G$ at $u_0$ if   $\delta_G(u_0)\ge 3$, $\delta_G(u_{\ell})=1$, and if $\ell\ge 2$, then $\delta_G(u_i)=2$ for each $i=1,\dots, \ell-1$.
Particularly, if $\ell=1$, then $u_0u_1$ is called  a pendant edge of $G$ at $u_0$.

\begin{lemma}\label{pendant}
Let $G$ be a graph with a pendant path $v_0v_1\dots v_{\ell}$ at $v_0$, where $\ell\ge 2$.
Then $C(G-v_iv_{i+1})>C(G-v_{i-1}v_i)$ for $i=1,\dots,\ell-1$.
\end{lemma}

\begin{proof}
Let $V_0=V(G)\setminus\{v_i:i=0,\dots, \ell\}$, $V_1=V_0\cup \{v_j:j=0,\dots,i-1\}$ and $V_2=\{v_j:j=i+1,\dots, \ell\}$.
As we pass from $G-v_{i-1}v_{i}$ to $G-v_iv_{i+1}$, the distance between any pair of vertices in $V_1$ and $V_2$ remains  unchanged.
Then
\begin{align*}
&\quad C(G-v_iv_{i+1})-C(G-v_{i-1}v_i)\\
&=2\sum_{u\in V_1}2^{-d_{G-v_iv_{i+1}}(u,v_i)}-2\sum_{u\in V_2}2^{-d_{G-v_{i-1}v_i}(u,v_i)}.
\end{align*}
Note that $V_0\subset V_1$ and for any $u\in V_0$,
\begin{align*}
d_{G-v_iv_{i+1}}(u,v_i)&=d_{G-v_iv_{i+1}}(u,v_0)+d_{G-v_iv_{i+1}}(v_0,v_i)\\
&=d_{G-v_iv_{i+1}}(u,v_0)+i.
\end{align*}
Then
\begin{align*}
&\quad C(G-v_iv_{i+1})-C(G-v_{i-1}v_i)\\
&=2\sum_{u\in V_0}2^{-d_{G-v_iv_{i+1}}(u,v_0)-i}+\sum_{j=0}^{i-1}2^{-j}-\sum_{j=0}^{\ell-i-1}2^{-j}\\
&=2^{-i+1}\sum_{u\in V_0}2^{-d_{G-v_iv_{i+1}}(u,v_0)}+(2-2^{-i+1})\\
&\quad  -(2-2^{-\ell+i+1})\\
&\ge 2^{-i+1}+(2-2^{-i+1})-(2-2^{-\ell+i+1})\\
&=2^{-\ell+i+1}\\
&>0.
\end{align*}
So the result follows.
\end{proof}


\section{Connectivity and edge connectivity}

The (vertex) connectivity $\kappa(G)$ of a  graph $G$ is defined as the minimum
number of vertices whose removal from $G$ results in a disconnected graph or
in the trivial graph. If $G$ is trivial or disconnected, then $\kappa(G)=0$. For a connected graph $G$ of order $n\ge 2$,  $1\le \kappa(G)\le n-1$, and
$\kappa(G)=n-1$ if and only if $G\cong K_n$.

\begin{theorem}\label{conn}
Let $G$ be an $n$-vertex graph with connectivity $k$, where $1\le k\le n-2$. Then
\[
R^L(G)\le \begin{cases}
\frac{2n^2-7n+13}{4}& \mbox{if $k=1$ and $n=5, \dots, 9$}  \\
\frac{n^2-3n+2}{2} & \mbox{if $k=1$ and $n=3, 4$ or $n\ge 9$}\\
\frac{n^2-2n+k}{2} & \mbox{if $k\ge 2$}
\end{cases}
\]
with equality if and only if $G\cong K_1\vee (K_2\cup K_{n-3})$ when $k=1$ and $5\le n\le 8$, $G\cong  K_1\vee(K_1\cup K_{n-2}), K_1\vee (K_2\cup K_{n-3})$ when  $k=1$ and $n=9$, and
$G\cong K_k\vee(K_1\cup K_{n-k-1})$ otherwise.
\end{theorem}

\begin{proof}	
Suppose that $G$ is an $n$-vertex  graph with connectivity $k$ that maximizes the link residual closeness.

By the definition of connectivity, there is a vertex subset
$V_0$ with $|V_0|=k$ so that
$G-V_0$ is disconnected.
Assume that $G_1$ is a component of $G-V_0$.
Let $V_1=V(G_1)$ and $V_2=V(G)-V_0-V_1$.
Let $n_i=|V_i|$ for $i=1,2$. Assume that $n_1\le n_2$.
As $n_1+n_2=n-k$, one has $n_1\le \lfloor\frac{n-k}{2} \rfloor$. By Lemma \ref{AddE}, $G[V_0\cup V_i]$ is complete for $i=1,2$.
That is,  $G\cong K_k\vee(K_{n_1}\cup K_{n-k-n_1})$.
By Lemma \ref{kst}, one has
\[
R^L(G)=\begin{cases}
\frac{n^2-3n+2}{2} & \mbox{if $k=n_1=1$},\\
f(n_1)& \mbox{if $k=1$ and $n_1\ge 2$},\\
g_k(n_1) & \mbox{otherwise},
\end{cases}
\]
where
\begin{align*}
f(x)
&=\frac{1}{4} (x^2+(n-1-x)^2)+\frac{1}{4}n^2-\frac{1}{4}n+\frac{1}{4}x-\frac{1}{2}\\
&=\frac{1}{2}\left(x^2-\left(n-\frac{3}{2}\right)x\right)+\frac{1}{4}(2n^2-3n-1)
\end{align*}
with $2\le x\le \lfloor\frac{n-1}{2}\rfloor$,
and
\begin{align*}
&\quad g_k(x)\\
&=\frac{1}{4}(x^2+(n-k-x)^2)+\frac{1}{4}(n^2-k^2)+\frac{k-1}{2}n-\frac{1}{2}\\
&=\frac{1}{2}(x^2-(n-k)x)+\frac{1}{2}(n^2-n-1)
\end{align*}
with $2\le x\le \lfloor\frac{n-k}{2}\rfloor$.
It is easy to see that $f(x)$ is strictly decreasing for $2\le x\le \frac{2n-3}{4}$, $f(\frac{n-1}{2})=f(\frac{n-2}{2})$ and
$g_k(x)$ is strictly  decreasing for $1\le x\le \lfloor\frac{n-k}{2}\rfloor$.

\noindent
{\bf Case 1.} $k=1$.

If $n_1\ge 3$, then $R^L(G)=f(n_1)<f(2)=R^L(K_1\vee(K_2\cup K_{n-3}))$, a contradiction. So $n_1=1,2$.

If $n=3,4$, then $G\cong K_1\vee (K_1\cup K_{n-2})$ with $R^L(G)=\frac{n^2-3n+2}{2}$.

Suppose that  $n\ge 5$.  Note that
$f(2) =\frac{2n^2-7n+13}{4}<  \frac{n^2-3n+2}{2}$ if and only if $n>9$.
and $f(2)=\frac{n^2-3n+2}{2}$ if and only if $n=9$.
Thus
\[
n_1=\begin{cases}
2 & \mbox{if  $n=5,\dots, 8$},\\
1,2 & \mbox{if $n=9$},\\
1 & \mbox{if $n\ge 10$}.
\end{cases}
\]
Correspondingly,
\[
G\cong K_2\vee(K_1\cup K_{n-3}) \mbox{ if }n=5,\dots, 8,
\]
\[
G\cong K_1\vee(K_1\cup K_{n-2}), K_1\vee(K_2\cup K_{n-3}) \mbox{ if }n=9,
\]
and
\[
G\cong K_1\vee(K_1\cup K_{n-2}) \mbox{ if }n\ge 10.
\]

\noindent
{\bf Case 2.} $k\ge 2$.

If $n_1\ge 2$, then $R^L(G)=g_k(n_1)<g_k(1)$,
a contradiction.
So  $n_1=1$,  $G\cong K_k\vee(K_1\cup K_{n-k-1})$ and $R^L(G)=\frac{n^2-2n+k}{2}$.
\end{proof}

 Suppose that  $G$ is an $n$-vertex  graph with connectivity at most $k$ that maximizes the link residual closeness, where $1\le k\le n-2$ and $n\ge 5$. By Lemma \ref{AddE}, $G$ is connected. So we have the following corollary from previous theorem.

\begin{corollary}\label{c1}
Let $G$ be an $n$-vertex  graph with connectivity at most $k$, where $1\le k\le n-2$ and $n\ge 5$. Then
\[
R^L(G)\le \begin{cases}
\frac{2n^2-7n+13}{4}& \mbox{if $k=1$ and $n=5, \dots, 9$}\\
\frac{n^2-3n+2}{2} & \mbox{if $k=1$ and $n\ge 9$} \\
\frac{n^2-2n+k}{2} & \mbox{if } k\ge 2
\end{cases}
\]
with equality if and only if $G\cong K_2\vee (K_1\cup K_{n-3})$ when $k=1$ and $n\le 8$, $G\cong  K_1\vee(K_1\cup K_{n-2}), K_2\vee(K_1\cup K_{n-3})$ when $k=1$ and $n=9$, and
$G\cong K_k\vee(K_1\cup K_{n-k-1})$ otherwise.
\end{corollary}

The edge connectivity $\kappa'(G)$ of a graph $G$ is defined as the minimum
number of edges whose removal from $G$ results in a disconnected graph or in
the trivial graph. If $G$ is trivial or disconnected, then $\kappa'(G)=0$. For a connected graph $G$ of order $n\ge 2$, $1\le \kappa'(G)\le n-1$ and $\kappa'(G)=n-1$ if and only if $G\cong K_n$.

\begin{lemma}\label{ks}
For positive integers $s$ and $n$ with $2s\le n$,
let $G$ be a graph obtained from $K_s\cup K_{n-s}$ by adding an edge.
Then \[
R^L(G)\le \frac{n^2-3n+2}{2}
\]
with equality if and only if $s=1$.
\end{lemma}

\begin{proof}
Let $u\in V(K_s)$ and $v\in V(K_{n-s})$ so that $uv\in E(G)$.
Let $V_1=V(K_s)\setminus\{u\}$ and  $V_2=V(K_{n-s})\setminus \{v\}$.
If $e=uv$, then
\[
C(G-e)=C(K_s)+C(K_{n-s}). 
\]
If $e=uw$ with $w\in V_1$, then we have either
$s=2$ and $C(G-e)=C(K_{n-2})+1+\frac{n-2}{2}>C(K_{n-2})+C(K_2)$
or $s\ge 3$ and
\begin{align*}
&\quad C(G-e)\\
&=C(K_s)-\frac{1}{2}+C(K_{n-s})+1+\frac{n-s-1}{2}+\frac{s-2}{2}\\
&\quad +\frac{1}{4}+\frac{(s-2)(n-s-1)}{4}+\frac{n-s-1}{8}\\
&=C(K_s)+C(K_{n-s})+\frac{s(n-s-1)}{4}+\frac{s-1}{2}+\frac{1}{4}\\
&\quad +\frac{n-s-1}{8} \\
&>C(K_s)+C(K_{n-s}).
\end{align*}
Similarly, if $e=vw$ with $w\in V_2$, then by direct calculation, we have
\[
C(G-e)>C(K_s)+C(K_{n-s}).
\]
If $e=wz$  with $w,z\in V_1$, then
\begin{align*}
&\quad C(G-e)\\
&=C(K_s)-\frac{1}{2}+C(K_{n-s})+1+\frac{n-s-1}{2}+\frac{s-1}{2}\\
&\quad +\frac{(s-1)(n-s-1)}{4}\\
&>C(K_s)+C(K_{n-s}).
\end{align*}
Similarly, if $e=wz$ with $w,z\in V_2$, then
\[
C(G-e)>C(K_s)+C(K_{n-s}).
\]
So
\begin{align*}
R^L(G)&=C(G-uv)=C(K_s)+C(K_{n-s})\\
&=s^2-ns+\frac{n^2-n}{2}.
\end{align*}
Let $f(s)=s^2-ns+\frac{n^2-n}{2}$.
Since $2s\le n$, $f(s)$ is strictly decreasing and hence $f(s)\le f(1)$ with equality if and only if $s=1$.
The result follows.
\end{proof}

\begin{theorem}\label{edgeconn}
Let $G$ be an $n$-vertex  graph with edge connectivity at most $r$, where $1\le r\le n-2$.
Then
\[
R^L(G)\le \begin{cases}
	\frac{n^2-3n+2}{2}&\mbox{ if }r=1\\
	\frac{n^2-2n+r}{2}&\mbox{ if }r\ge 2
\end{cases}
\]
with equality if and only if $G\cong K_r\vee(K_1\cup K_{n-r-1})$.
\end{theorem}

\begin{proof}
Suppose that $G$ is an $n$-vertex graph with edge connectivity at most $r$ that maximizes the link residual closeness.

By Lemma \ref{AddE}, $G$ is connected.

If $r=1$, then $G$ has a (cut) edge, say $e$, and we have by Lemma \ref{AddE} that
 $G-e\cong K_s\cup K_{n-s}$ for some $s\ge 1$,  so it follows from
Lemma \ref{ks} that $G\cong K_1\vee(K_1\cup K_{n-2})$ with $R^L(G)=\frac{n^2-3n+2}{2}$.

Suppose that $r\ge 2$. Then $n\ge 4$.
Let $\kappa(G)=k$.
Then $k\le r$. We claim that $k=r$.
Suppose to the contrary that $k<r$.

If $k=1$ and $n=4$, then $r=2$, and it is easy to see that there is no such graphs. So $n\ge 5$ if $k=1$.  Suppose first that $k=1$ and  $n=5,\dots, 8$.
Let $H=K_1\vee (K_2\cup K_{n-3})$.
By Theorem \ref{conn}, $R^L(G)\le R^L(H)$.
Let $u_1,u_2$ be the vertices of $K_2$ (the $1$st inner copy of $H$).
Let $W$ be  a subset or vertices of $K_{n-3}$ (the $2$nd inner copy of $H$) with $|W|=r-2$ if $r\ge 3$.
Let $H'=H+\{u_1w:w\in V(K_{n-3}) \}$ if $r=2$, and
$H'=H+\{u_1w:w\in V(K_{n-3}) \}+\{u_2w:w\in W \}$ if $r\ge 3$.
Evidently, $u_2$ is the only vertex of $H'$ with minimum degree $r$ and each other vertex has degree at least $n-2$.
It is easy to see that $\kappa'(H')=r$.
By Lemma \ref{AddE}, $R^L(H')>R^L(H)\ge R^L(G)$, a contradiction.

Suppose next that $k=1$ and $n\ge 9$ or $k\ge 2$.
Note that $1\le r-k\le n-k-2$.
Let $H=K_k\vee(K_1\cup K_{n-k-1})$.
By Theorem \ref{conn}, $R^L(G)\le R^L(H)$.
Let $u$ be the vertex with degree $k$  and $W$ be a subset of vertices of $K_{n-k-1}$ (the $2$nd inner copy of $H$) with $|W|=r-k$.
Let $H'=H+\{uw: w\in W \}$. Evidently, $u$ is the only vertex of $H'$ with minimum degree $r$ and each other vertex has degree at least $n-2$.
It is easy to see that $\kappa'(H')=r$.
By Lemma \ref{AddE}, $R^L(H')>R^L(H)\ge R^L(G)$, also a contradiction.

Now we conclude that $k=r$, so by Theorem \ref{conn}, we have
$G\cong K_r\vee(K_1\cup K_{n-r-1})$ with $R^L(G)=\frac{n^2-2n+r}{2}$.
\end{proof}

Noting that $\kappa'(G)\le \delta(G)$ for a connected graph $G$ with minimum degree $\delta (G)$ and by similar argument as in Theorem \ref{edgeconn}, we have the following result.

\begin{corollary}
Let $G$ be an $n$-vertex graph with minimum degree at most $\delta$, where $1\le \delta\le n-2$.
Then
\[
R^L(G)\le \begin{cases}
	\frac{n^2-3n+2}{2}&\mbox{ if }\delta=1\\
	\frac{n^2-2n+\delta}{2}&\mbox{ if }\delta\ge 2
\end{cases}
\]
with equality if and only if $G\cong K_\delta\vee (K_1\cup K_{n-\delta-1})$.
\end{corollary}

\section{Bipartiteness}

The bipartiteness of a graph $G$ is the minimum number of vertices whose deletion from $G$  yields a bipartite graph.

\begin{theorem} \label{adj}
Let $G$ be a bipartite graph on $n\ge 3$ vertices. Then
\[
R^L(G)\le \frac{n^2-n-3}{4}+\frac{1}{2}\left \lfloor \frac{n^2}{4}\right\rfloor
\]
with equality if and only if $G\cong K_{\lfloor n/2\rfloor, \lceil n/2\rceil}$.
\end{theorem}

\begin{proof}
Let $r$ and $s$ be the partite sizes of $G$ with  $1\le r\le s$.
Note that
\begin{align*}
R^L(K_{1,s})&=C(K_{1, s-1})=\frac{(s-1)(s+2)}{4}=\frac{n^2-n-2}{4}\\
&<\frac{n^2-n-3}{4}+\frac{1}{2}\left \lfloor \frac{n^2}{4}\right\rfloor.
\end{align*}
and if $r\ge 2$, then
\begin{align*}
R^L(K_{r,s})&=C(K_{r,s})-2\times \frac{1}{2}+2\times \frac{1}{8}\\
&=\frac{1}{2}{r\choose 2}+\frac{1}{2}{s\choose 2}+rs-\frac{3}{4}\\
&=\frac{r(r-1)}{4}+\frac{s(s-1)}{4}+rs-\frac{3}{4}\\
&=\frac{n^2-n-3}{4}+\frac{rs}{2}\\
&\le \frac{n^2-n-3}{4}+\frac{1}{2}\left \lfloor \frac{n^2}{4}\right\rfloor
\end{align*}
with equality if and only if $r=s$. Now the result follows from Lemma \ref{AddE}.
\end{proof}

\begin{theorem}
Let $G$ be an $n$-vertex graph with bipartiteness $k$, where $1\le k\le n-2$.
Then
\begin{align*}
R(G)& \le \frac{3}{8}n^2-\frac{1}{4}n+\frac{1}{4}nk-\frac{k^2}{8}-\frac{k}{4}\\
&\quad -\begin{cases}
\frac{5}{8}&\mbox{ if  $n-k$ is odd}\\[2mm]
\frac{1}{2}&\mbox{ otherwise}
\end{cases}
\end{align*}
with equality if and only $G\cong K_k\vee K_{\lfloor (n-k)/2 \rfloor,\lceil (n-k)/2\rceil}$.
\end{theorem}
\begin{proof}
Suppose that $G$ is an $n$-vertex graph with bipartiteness $k$ that maximizes the link residual closeness.

Let $V_0\subset V(G)$ with $|V_0|=k$ such that $G-V_0$ is bipartite.
Let $V_1,V_2$ be the bipartite sets of $G-V_0$.
By Lemma \ref{AddE}, $G-V_0$ is a complete bipartite graph and $G[V_0\cup V_i]$ is complete for $i=1,2$.
Let $r=|V_1|$ and $s=|V_2|$.
Then $r+s=n-k$ and
we may assume that $G=K_k\vee K_{r,s}$.
 Let $wz\in E(G)$. Then $d_G(w,z)=1$, $d_{G-wz}(w,z)=2$, and
if $\{u,v\}\ne \{w,z\}$, then $d_{G-wz}(u,v)=d_{G}(u,v)$. So
\begin{align*}
R^L(G)&=C(G)-2\times \frac{1}{2}+2\times \frac{1}{4}\\
&={k\choose 2}+\frac{1}{2}{r\choose 2}+\frac{1}{2}{s\choose 2}+k(r+s)+rs-\frac{1}{2}\\
&=k(n-k)+\frac{2k(k-1)+(n-k)(n-k-1)}{4}\\
&\quad +\frac{r(n-k)-r^2}{2}-\frac{1}{2}.
\end{align*}
Denote by $g_k(r)$ the above expression for $R^L(G)$.
Assume that $r\le s$. Then $r\le \lfloor\frac{n-k}{2} \rfloor$.
It is easy to see that $g_k(r)$ is strictly increasing  for  $1\le r\le \lfloor\frac{n-k}{2} \rfloor$.
So $r=\lfloor\frac{n-k}{2} \rfloor$. This is because, if $r< \lfloor\frac{n-k}{2} \rfloor$, then
\begin{align*}
R^L(G)&=g_k(r)<g_k\left(\left\lfloor\frac{n-k}{2} \right\rfloor\right)\\
&=R^L(K_k\vee K_{\lfloor (n-k)/2 \rfloor,\lceil (n-k)/2 \rceil}),
\end{align*}
which is a contradiction.
So $G\cong K_k\vee K_{\lfloor (n-k)/2\rfloor,\lceil (n-k)/2 \rceil}$ with $R^L(G)=g_k\left(\left\lfloor\frac{n-k}{2} \right\rfloor\right)$.
\end{proof}

\section{Independence number and matching number}

The independence number of a graph $G$, denoted by $\alpha(G)$, is the maximum cardinality of an independent set of vertices in $G$. Evidently, the complete graph is the unique one with independence number one.

\begin{theorem}
Let $G$ be an $n$-vertex graph with independent number $\alpha$, where $2\le \alpha\le n-1$. Then
\[
R^L(G)\le
\begin{cases}
\frac{n^2-n-2}{4} & \mbox{if $\alpha=n-1$}\\
\frac{n^2-n-1}{2}-\frac{\alpha^2-\alpha}{4} & \mbox{if $\alpha\le n-2$}
\end{cases}
\]
with equality if and only if $G\cong K_{n-\alpha}\vee \overline{K_\alpha}$.
\end{theorem}

\begin{proof} Suppose that $G$ is an $n$-vertex graph with independence number $\alpha$ that maximizes the link residual closeness.

Let $S$ be an independent set of $G$ with $|S|=\alpha$. By Lemma \ref{AddE}, $G[V(G)\setminus S]\cong K_{n-\alpha}$, so
$G\cong K_{n-\alpha}\vee \overline{K_\alpha}\cong K_{n-\alpha}\vee (\underbrace{K_1\cup \dots\cup K_1}_{\alpha \mbox{ times}})$.
By Lemma \ref{kst},
\[
R^L(G)=\begin{cases}
\frac{n^2-n-2}{4} & \mbox{if $\alpha=n-1$,}\\
\frac{n^2-n-1}{2}-\frac{\alpha^2-\alpha}{4} & \mbox{if $\alpha\le n-2$,}
\end{cases}
\]
as desired.
\end{proof}

The matching number of a graph $G$ is the cardinality of a maximum matching of $G$, denoted by $\beta(G)$. Berge \cite{Be} extended
Tutte's 1-Factor Theorem \cite{Tu} to  the Berge-Tutte Formula for
the  matching number of a graph, which states that
\[
\beta(G)=\frac{1}{2}\min_{S\subset V(G)} (|V(G)|+|S|-o(G-S)),
\]
where, for a graph $H$,
$o(H)$ denotes the number of odd components (those with an odd number of vertices) of $H$.

Let $G$ be a connected graph on $n\ge 3$ vertices with matching number one.
If $n=3$, then  $G\cong S_3, K_3$ with $R^L(S_3)=1<R^L(K_3)=\frac{5}{2}$.
If $n\ge 4$, then   $G\cong S_n$ with  $R^L(S_n)=\frac{(n-2)(n+1)}{4}$.

By \cite[Theorem 6.1]{ZLG}, $K_n$ maximizes the link residual closeness among all $n$-vertex connected graphs and  $R^L(K_n)=\frac{n^2-n-1}{2}$. Thus, if $G$ is a connected graph on $n\ge 2$ vertices with matching number $\lfloor\frac{n}{2}\rfloor$, then  $R^L(G)\le \frac{n^2-n-1}{2}$ with equality if and only if $G\cong K_n$.

\begin{theorem}\label{matching}
Let $G$ be an $n$-vertex graph with matching number $\beta$, where $2\le \beta\le \lfloor\frac{n}{2}\rfloor-1$.

If $2\le \beta< \frac{2n+3}{5}$, then
\[
R^L(G)\le
\frac{1}{4}n^2-\frac{1}{4}n-\frac{1}{4}\beta^2+\frac{1}{2}n\beta-\frac{1}{4}\beta-\frac{1}{2}
\]
with equality if and only if $G\cong K_\beta\vee \overline{K_{n-\beta}}$.

If $\frac{2n+3}{5}< \beta <\frac{n}{2}-1$, then
\[
R^L(G)\le \frac{1}{4}n^2+\frac{3}{4}n+\beta^2-\frac{7}{2}\beta+\frac{1}{2}
\]
with equality if and only if $G\cong K_2\vee \left((n-2\beta+1)K_1\cup K_{2\beta-3}\right)$.

If $\beta=\frac{2n+3}{5}$, then
\[
R^L(G)\le \frac{1}{4}n^2-\frac{1}{4}n-\frac{1}{4}\beta^2+\frac{1}{2}n\beta-\frac{1}{4}\beta-\frac{1}{2}
\]
with equality if and only if $G\cong K_\beta\vee \overline{K_{n-\beta}}, K_2\vee \left((n-2\beta+1)K_1\cup K_{2\beta-3}\right)$.

If $\beta=\frac{n}{2}-1$, then
\[
R^L(G)\le 2\beta^2+\beta-\frac{1}{2}
\]
with equality if and only if $G\cong K_{2\beta+1}\cup (n-2\beta-1)K_1$.
\end{theorem}

\begin{proof}
Suppose that $G$ is an $n$-vertex graph with matching number $\beta$ that maximizes the link residual closeness.

By the  Tutte-Berge formula, there is a vertex set $S\subset V(G)$ such that $\beta=\frac{1}{2}(n+|S|-o(G-S))$.
Let $s=|S|$ and $t=o(G-S)$.
Then $n-2\beta  =t-s$.
Since $\beta \le \lfloor\frac{n}{2}\rfloor-1$, we have $t-s\ge 2$ and hence $t\ge 2$.
As $n-s\ge t=n-2\beta +s$, we  have $s\le \beta$.

\noindent
{\bf Claim 1.} All components of $G-S$ are odd.

Otherwise, there is an even component in $G-S$.
Let $H$ be a graph obtained from $G$ by adding all possible edges between vertices in an even component  and an odd component of $G-S$.
Then
$\beta(H)\ge\beta(G)=\beta$.
Moreover, $o(H-S)=t$ and $\beta(H)\le \frac{1}{2}(n+s-t)=\beta$.
So $\beta(H)=\beta$.
By Lemma \ref{AddE}, $R^L(H)>R^L(G)$, a contradiction.
This proves Claim 1.

Denote by $G_i$ with $i=1,\dots, t$ the components of $G-S$ with.
 $n_i=|V(G_i)|$.
By Claim 1 and Lemma \ref{AddE}, $n_i$ is odd and $G_i\cong K_{n_i}$ for $i=1,\dots,t$.
So $G\cong K_s\vee (K_{n_1}\cup \dots\cup K_{n_t})$, where $K_0\vee (K_{n_1}\cup \dots\cup K_{n_t})
=K_{n_1}\cup \dots\cup K_{n_t}$.
Assume that $n_1\le\dots\le n_t$.

\noindent
{\bf Claim 2.}  $n_1=\dots=n_{t-1}=1$ and $n_t=n-s-t+1$.

Otherwise, $n_{t-1}\ge 3$.
Let
$G'=K_1\vee ((t-1)K_1\cup K_{n-t})$ if $s=1$ and  $n_1=\dots=n_{t-1}=3$, and $G'=K_s\vee (K_{n_1}\cup\dots\cup K_{n_{t-2}}\cup K_{n_{t-1}-2}\cup K_{n_t+2})$ otherwise.
As $o(G'-S)=t$, $\beta(G')\le \frac{n+s-t}{2}=\beta$.
Note that $\beta(G')\ge \frac{\sum_{i=1}^tn_i-t}{2}+s=\beta$.
So $\beta(G')=\beta$.

Suppose first that $s=0$. Note that $n_i=1$ or $n_i\ge 3$, and  $n_{t-1}\ge 3$. Then it is easy to see
\[
R^L(G)=\sum_{i=1}^t C(K_{n_i})-\frac{1}{2},
\]
\[
R^L(G')=\sum_{i=1}^{t-2} C(K_{n_i})+C(K_{n_{t-1}-2})+C(K_{n_t+2})-\frac{1}{2},
\]
so
\begin{align*}
&\quad R^L(G')-R^L(G)\\
&=C(K_{n_{t-1}-2})+C(K_{n_t+2})-C(K_{n_{t-1}})-C(K_{n_t})\\
&=\frac{1}{2}\left((n_{t}+2)^2+(n_{t-1}-2)^2\right)-\frac{1}{2}\left(n_{t}^2+n_{t-1}^2\right)\\
&=2(n_{t}-n_{t-1}+2)\\
&>0,
\end{align*}
a contradiction.
Suppose next that  $s=1$ and $n_1=\dots=n_{t-1}=3$.  Then  $n_t=n-3t+2$, so by Lemma \ref{kst}, we have
\[
R^L(G)=\frac{9}{4}(t-1)+\frac{1}{4}(n-3t+2)^2+\frac{1}{4}n^2-\frac{1}{4}n+\frac{1}{4}
\]
and
\[
R^L(G')=\frac{1}{4}(t-2)+\frac{1}{4}(n-t)^2+\frac{1}{4}n^2-\frac{1}{2}n.
\]
As $n\ge 3t+1$, we have
\begin{align*}
&\quad R^L(G')-R^L(G)\\
&=\frac{1}{4}(t-2)+\frac{1}{4}(n-t)^2-\frac{1}{4}n
-\frac{9}{4}(t-1)-\frac{1}{4}(n-3t+2)^2-\frac{1}{4}\\
&=
-2t^2+(n+1)t-\frac{5}{4}n+\frac{1}{2}\\
&\ge
-2t^2+\left(t-\frac{5}{4}\right)(3t+1)+t+\frac{1}{2}\\
&=
t^2-\frac{7}{4}t-\frac{3}{4}\\
&>0,
\end{align*}
a contradiction. Now we are left with the following three cases:
(a) $s=1$ and $n_1=\dots =n_r=1$, $n_{r+1}\ge 3$ for some $r=1,\dots, t-2$, or
(b) $s=1$, $n_1\ge 3$ and $n_{t-1}\ge 5$, or
(c) $s\ge 2$.
For any case, we have by Lemma \ref{kst} that
\begin{align*}
R^L(G')-R^L(G)
&=\frac{1}{4}((n_t+2)^2+(n_{t-1}-2)^2)-\frac{1}{4}(n_t^2+n_{t-1}^2)\\
&=n_t-n_{t-1}+2\\
&>0,
\end{align*}
also a contradiction. This proves Claim 2.

Recall that $t=n-2\beta+s$. By Claim 2,
 $G\cong K_s\vee ((n-2\beta+s-1)K_1\cup K_{2\beta-2s+1})$.
For fixed $n$, let $f(s)=R^L(K_s\vee ((n-2\beta+s-1)K_1\cup K_{2\beta-2s+1}))$. We want to maximize $f(s)$.

Suppose that $s\ge 2$. By Lemma \ref{kst} again,
\begin{align*}
f(s)&=\frac{1}{4}(2\beta-2s+1)^2+\frac{1}{4}(n-2\beta+s-1)\\
&\quad +\frac{1}{4}n^2-\frac{1}{4}s^2+\frac{1}{2}ns-\frac{1}{2}n-\frac{1}{2}.
\end{align*}
As a quadratic function of $s\in [2,\beta]$, $f(s)$  has axis of symmetry $s=s^*:=\frac{4\beta}{3}-\frac{n}{3}+\frac{1}{2}$.
Note that $\beta-s^*>s^*-2$ if and only if $\beta<\frac{2n+3}{5}$.
Thus,
\[
f(\beta)>f(s) \mbox{ if $s\in [2,\beta)$ and $\beta<\frac{2n+3}{5}$},
\]
\[
f(2)=f(\beta)>f(s)  \mbox{ if  $s\in(2,\beta)$ and  $\beta=\frac{2n+3}{5}$},
\]
and
\[
f(2)>f(s)  \mbox{ if $s\in (2,\beta]$ and $\beta>\frac{2n+3}{5}$}.
\]
That is, $f(s)$ is maximized if and only if
\[
s=\begin{cases}
\beta&\mbox{ if }\beta<\frac{2n+3}{5},\\
2,\beta&\mbox{ if }\beta=\frac{2n+3}{5},\\
2&\mbox{ if }\beta>\frac{2n+3}{5},
\end{cases}
\]
and correspondingly, the maximum of $f(s)$ is equal to
$f(\beta)=\frac{1}{4}n^2-\frac{1}{4}n-\frac{1}{4}\beta^2+\frac{1}{2}n\beta-\frac{1}{4}\beta-\frac{1}{2}$ if $\beta\le \frac{2n+3}{5}$, and
$f(2)=\frac{1}{4}n^2+\frac{3}{4}n+\beta^2-\frac{7}{2}\beta+1$  if  $\beta\ge \frac{2n+3}{5}$.

Next, by Lemma \ref{kst}, we have
\[
f(1)=\frac{1}{4}(2\beta-1)^2+\frac{1}{4}(n-2\beta-1)+\frac{1}{4}n^2-\frac{1}{2}n.
\]
For $\beta\in [2,\frac{2n+3}{5}]$, we have
\[
f(\beta)-f(1)=\phi(\beta):= -\frac{5}{4}\beta^2+\frac{1}{2}n\beta+\frac{5}{4}\beta-\frac{1}{2}.
\]
As
\begin{align*}
\phi(\beta) & \ge \phi\left(\frac{2n+3}{5}\right)\\
&=-\frac{(2n+3)^2}{20}+\frac{(2n+3)n}{10}+\frac{2n+3}{4}-\frac{1}{2}\\
&=\frac{n-1}{5}>0,
\end{align*}
 we have
$f(\beta)>f(1)$.
On the other hand, we have $f(2)>f(1)$ as $f(2)-f(1)=n-2\beta+1>0$. Thus $s=1$ is impossible.

Now,  we have by direct calculation that
\[
f(0)=
2\beta^2+\beta-\frac{1}{2}.
\]
For $\beta\in [2,\frac{2n+3}{5}]$, we have
\[
f(\beta)-f(0)=\varphi(\beta):=\frac{1}{4}n^2-\frac{1}{4}n-\frac{9}{4}\beta^2+\frac{1}{2}n\beta-\frac{5}{4}\beta.
\]
and
$\varphi(\beta)\ge g\left(\frac{2n+3}{5}\right)=\frac{9n^2+63n-6}{100}>0$,
so $f(\beta)>f(0)$, implying that $s=0$ is impossible.
On the other hand, for $\beta\in [\frac{2n+3}{5},\lfloor \frac{n}{2}\rfloor-1]$,
\[
f(2)-f(0)=\psi(\beta):=\frac{1}{4}n^2+\frac{3}{4}n-\beta^2-\frac{9}{2}\beta+\frac{3}{2}.
\]
As
$\psi(\beta)=0$ has a positive root \[\beta_1=\frac{\sqrt{(n+\frac{3}{2})^2+24}}{2}-\frac{9}{4}>\frac{n}{2}-\frac{3}{2},\]
we have
$f(2)>f(0)$
if $n$ is odd, or if $\beta<\frac{n}{2}-1$ and $n$ is even. In such cases, $s=0$ is also impossible. If
$\beta=\frac{n}{2}-1$, then  $f(2)-f(0)=-2n+8<0$, so $f(2)<f(0)$, implying that $s=2$ is impossible.

Therefore, we conclude that $f(s)$ is maximized if and only if
\[
s=\begin{cases}
\beta&\mbox{ if }2\le \beta<\frac{2n+3}{5},\\
2&\mbox{ if }\frac{2n+3}{5}<\beta<\frac{n}{2}-1,\\
2,\beta&\mbox{ if }\beta=\frac{2n+3}{5},\\
0 &\mbox{ if }\beta=\frac{n}{2}-1.
\end{cases}
\]
Correspondingly,
\[
G\cong K_\beta\vee \overline{K_{n-\beta}}
\]
with
$R^L(G)=f(\beta)$
if $2\le \beta<\frac{2n+3}{5}$,
\[
G\cong K_2\vee \left((n-2\beta+1)K_1\cup K_{2\beta-3}\right)
\]
with $R^L(G)=f(2)$
if $\frac{2n+3}{5}<\beta<\frac{n}{2}-1$,
\[
G\cong K_\beta\vee \overline{K_{n-\beta}},K_2\vee \left((n-2\beta+1)K_1\cup K_{2\beta-3}\right)
\]
with $R^L(G)=f(\beta)=f(0)$
if $\beta=\frac{2n+3}{5}$,
and
\[
G\cong K_{2\beta+1}\cup (n-2\beta-1)K_1
\]
with $R^L(G)=f(0)$
if $\beta=\frac{n}{2}-1$.
This completes the proof.
\end{proof}

Suppose that $G$ is connected in  the above proof. As $t-s=n-2\beta\ge 2$, we have $s\ne 0$, so one immediately has the following corollary.

\begin{corollary}\label{ma}
Let $G$ be an $n$-vertex connected graph with matching number $\beta$, where $2\le \beta\le \lfloor\frac{n}{2}\rfloor-1$.

If $2\le \beta< \frac{2n+3}{5}$, then
\[
R^L(G)\le
\frac{1}{4}n^2-\frac{1}{4}n-\frac{1}{4}\beta^2+\frac{1}{2}n\beta-\frac{1}{4}\beta-\frac{1}{2}
\]
with equality if and only if $G\cong K_\beta\vee \overline{K_{n-\beta}}$.

If $\frac{2n+3}{5}< \beta \le\frac{n}{2}-1$, then
\[
R^L(G)\le \frac{1}{4}n^2+\frac{3}{4}n+\beta^2-\frac{7}{2}\beta+\frac{1}{2}
\]
with equality if and only if $G\cong K_2\vee \left((n-2\beta+1)K_1\cup K_{2\beta-3}\right)$.

If $\beta=\frac{2n+3}{5}$, then
\[
R^L(G)\le \frac{1}{4}n^2-\frac{1}{4}n-\frac{1}{4}\beta^2+\frac{1}{2}n\beta-\frac{1}{4}\beta-\frac{1}{2}
\]
with equality if and only if $G\cong K_\beta\vee \overline{K_{n-\beta}}, K_2\vee \left((n-2\beta+1)K_1\cup K_{2\beta-3}\right)$.
\end{corollary}

\section{Chromatic number}

The chromatic number of a graph $G$, written $\chi (G)$, is the least $k$ such
that the vertex set can be expressed as the union of $k$ independent sets.

The chromatic number of a graph $G$ is defined to be the number of colors necessary to color the vertices of $G$  such that no two adjacent vertices have the same color.

\begin{theorem}
For $n\ge 3$ and $2\le k\le n$,
let $G$ be an $n$-vertex  graph with chromatic number $k$.
Suppose that $n=kq+r$ with $0\le r<k$.
Then
\[
R^L(G)\le\begin{cases}
\frac{n^2-n-3}{4}+\frac{1}{2}\left\lfloor\frac{1}{4}n^2\right\rfloor & \mbox{if }k=2\\[2mm]
\frac{n^2-1}{2}-\frac{n+r(q+1)^2+(k-r)q^2}{4} & \mbox{if } k\ge 3
\end{cases}
\]
with equality if and only if $G$ is a balanced complete $k$-partite graph.
\end{theorem}
\begin{proof}
Suppose that $G$ is an $n$-vertex graph with chromatic number $k$ that maximizes the link residual closeness.

Denote by $V_1,\dots, V_k$ the coloring classes of $G$.
Let $n_i=|V_i|$ for $i=1,\dots,k$.
Assume that $n_i\ge n_{i+1}$ for $i=1,\dots,k-1$.
By Lemma \ref{AddE}, $G\cong K_{n_1,\dots,n_k}$.

If $k=2$, then the result follows from Theorem \ref{adj}.

Suppose next that $k\ge 3$.
Recall that $\sum_{i=1}^kn_i=n$.
For any $wz\in E(G)$, $d_{G}(w,z)=1$, $d_{G-wz}(w,z)=2$, and if $\{u,v\}\ne \{w,z\}$, then $d_{G-wz}(u,v)=d_G(u,v)$. So
\begin{align*}
 R^L(G)
&=C(G)-\frac{1}{2}\\
&=\sum_{i=1}^k\frac{1}{4}n_i(n_i-1)+\sum_{1\le i<j\le k}n_in_j-\frac{1}{2}\\
&=\frac{1}{2}n^2-\frac{1}{4}n-\frac{1}{4}\sum_{i=1}^{k}n_i^2-\frac{1}{2}.
\end{align*}

Suppose  that $n_1-n_k\ge 2$.
Let $m_1=n_1-1$, $m_k=n_k+1$ and $m_i=n_i$ for $i=2, \dots, k-1$. Let  $G'=K_{m_1,\dots,m_k}$.
Then
\begin{align*}
R^L(G')-R^L(G)&=-\frac{1}{4}(m_1^2+m_k^2)+\frac{1}{4}(n_1^2+n_k^2)\\
&=\frac{1}{2}(n_1-n_k-1)\\
&>0,
\end{align*}
a contradiction. So $n_1-n_k=0,1$.
As $n=kq+r$ with $0\le r<k$, we have $n_i=q+1$ for $i$ with $1\le i\le r$ and $n_i=q$ for $i=r+1,\dots,k$.
Then
\[
R^L(G)=\frac{1}{2}n^2-\frac{1}{4}n-\frac{1}{4}r(q+1)^2-\frac{1}{4}(k-r)q^2-\frac{1}{2},
\]
as desired.
\end{proof}

\section{Number of cut edges}


A cut edge (or bridge) is a single edge whose removal disconnects a graph.
For  a connected graph $G$ on $n$ vertices, if $G$ is a tree, then
every edge is a cut edge. Otherwise, there are at most $n-3$ cut edges.

A path $u_0u_1\dots u_s$ in a graph $G$  is
an internal path of $G$ with length $s$ if $\delta_G(u_0), \delta_G(u_s)\ge 3$ and if $s\ge 2$, then $\delta_G(u_i)=2$ for each $i=1,\dots, s-1$.

\begin{lemma}\label{new1}
Let $G$ be a connected graph with an internal path $u_0\dots u_t$, where $t\ge 1$.
Let $N_0=N_G(u_t)\setminus\{u_{t-1}\}$ and
\[
G'=G-\{u_tw:w\in N_0 \}+\{u_0w:w\in N_0\}.
\]
Then $R^L(G')>R^L(G)$.
\end{lemma}

\begin{proof} Let $G_1$ be the component of $G-u_0u_1$ containing $u_0$ and $G_2$ the component of
$G-u_{t-1}u_t$ containing $u_t$.
Let $V_i=V(G_i)$ for $i=1,2$.
Let $G_2'=G'[V_2\setminus\{u_t\}\cup\{u_0\}]$. Then $G_2'\cong G_2$.
Assume that $R^L(G')=C(G'-f)$ with $f\in E(G')$. By Lemma \ref{pendant},
$f\ne u_iu_{i+1}$ for $1\le i\le t-1$.
If $f=u_0u_1$, then
\begin{align*}
&\quad \sum_{w\in V_1} 2^{-d_{G'}(z,w)}-\sum_{i=1}^t2^{-d_G(z,u_i)}\\
&=2^{-d_{G}(z,u_t)} \left(\sum_{w\in V_1}2^{-d_G(u_0,w)}-\sum_{i=1}^t2^{-(t-i)} \right)\\
&\ge 2^{-d_{G}(z,u_t)} \left(1+2\cdot\frac{1}{2}-(2-2^{-(t-1)}) \right)\\
&>0
\end{align*}
for any $z\in V_2\setminus\{u_t\}$,
so
\begin{align*}
&\quad R^L(G')-R^L(G)\\
&\ge C(G'-f)-C(G-f)\\
&=2\sum_{z\in V_2\setminus\{u_t\}}\left(\sum_{w\in V_1} 2^{-d_{G'}(z,w)}-\sum_{i=1}^t2^{-d_G(z,u_i)}\right)\\
&>0.
\end{align*}
If $f\in E(G_1)\cup E(G_2')$, then
\[
\sum_{i=0}^{t}\left(2^{-d_{G'-f}(z,u_i)}-2^{-d_{G-f}(z,u_i)}\right)=0
\]
for any $z\in V_2\setminus\{u_t\}$, so
\begin{align*}
&\quad R^L(G')-R^L(G)\\
&\ge C(G'-f)-C(G-f)\\
&=2\sum_{z\in V_2\setminus\{u_t\}}\sum_{w\in V_1\setminus\{u_0\}}\left(2^{-d_{G'-f}(z,w)}-2^{-d_{G-f}(z,w)} \right)\\
&>0,
\end{align*}
where the second inequality follows because
$d_{G'-f}(z,w)=d_{G-f}(z,w)-t<d_{G-f}(z,w)$ or
$d_{G'-f}(z,w)=d_{G-f}(z,w)=\infty$
for any $z\in V_2\setminus\{u_t\}$ and $w\in V_1$, and there is indeed such vertex pair for which the former inequality holds.
It follows that $R^L(G')>R^L(G)$.
\end{proof}

\begin{lemma}\label{new2}
Let $G$ be a connected graph with   a
pendant path $u_0u_1\dots u_t$ at $u_0$, where $t\ge 2$.
Let
\[
G'=G-\{u_{i-1}u_{i}:i=2,\dots,t\}+\{u_0u_i:i=2,\dots,t\}.
\]
Then $R^L(G')>R^L(G)$.
\end{lemma}

\begin{proof}
Let $V_1=V(G)\setminus\{u_1,\dots,u_t\}$.
Assume that $R^L(G')=C(G'-e)$ with $e\in E(G')$.

Suppose first that $e\in E(G)\setminus\{u_0u_1\}$.
As we pass from $G-e$ to $G'-e$, the distance between any pair of  vertices in $V_1\cup \{u_1\}$ remains unchanged. For any $w\in V_1$ and $i=2,\dots,t$,
$d_{G'-e}(w,u_i)=d_{G-e}(w,u_i)-i+1$ if $d_{G-e}(w,u_0)$ is finite, and there is indeed such a vertex $w=w_0$, a neighbor of $u_0$ that is not incident to $e$.
Moreover, $G[\{u_0,\dots,u_t\}]\cong P_{t+1}$ and $G'[\{u_0,\dots,u_t\}]\cong S_{t+1}$.
Note that $C(S_{t+1})\ge C(P_{t+1})$.
So
\begin{align*}
&\quad R^L(G')-R^L(G)\\
&\ge C(G'-e)-C(G-e)\\
&=2\sum_{i=2}^t\sum_{w\in V_1}\left(2^{-d_{G'-e}(u_i,w)}-2^{-d_{G-e}(u_i,w)} \right)\\
&\quad +C(S_{t+1})-C(P_{t+1})\\
&\ge 2\sum_{i=2}^t\left(2^{-d_{G'-e}(u_i,w_0)}-2^{-d_{G-e}(u_i,w_0)} \right)\\
&=2\sum_{i=2}^t\left(2^{-2}-2^{-(i+1)} \right)\\
&>0.
\end{align*}
If $e=u_0u_i$ for some $i=1,\dots,t$, say $i=1$, then
\begin{align*}
&\quad R^L(G')-R^L(G)\\
&\ge C(G'-e)-C(G-e)\\
&=2\sum_{i=2}^t\sum_{w\in V_1\setminus\{u_0\}} 2^{-1-d_{G}(u_0,w)}+C(S_{t})-C(P_{t})\\
&>0.
\end{align*}
Thus, in either case, we have  $R^L(G')>R^L(G)$.
\end{proof}

If $G$ is a tree on $n\ge 2$ vertices, then $R^L(G)\le \frac{(n-2)(n+1)}{4}$ with equality if and only if $G\cong S_n$, which is known in \cite{ZLG} and follows also from Lemma \ref{new2}.

For $1\le k\le n-3$, by $C_{n,k}$, we denote the $n$-vertex graph obtained from $K_{n-k}$ by attaching $k$ pendant edges at one common vertex. It is clear that $C_{n,1}\cong K_1\vee (K_1\cup K_{n-2})$. For $n\ge 5$, let $C_{n,2}'$ be be graph obtained from $C_{n-1, 1}$ by attaching a pendant edge at a vertex of degree $n-3$.

\begin{theorem} \label{cutedges}
Let $G$ be an $n$-vertex  connected graph with $k$ cut
edges, where $1\le k\le n-3$. Then
\[
R^L(G)\le\frac{n^2-nk}{2}-n+\frac{k^2+3k}{4}
\]
with equality if and only if $G\cong C_{n,k}$ for $k\ne 2$ and  $G\cong C_{n,k}, C_{n,k}'$ for $k=2$.
\end{theorem}

\begin{proof}
If $k=1$, then the edge connectivity of $G$ is $1$, so
the result follows from Theorem \ref{edgeconn}.

Suppose that $k\ge 2$. Let $G$ be a  connected graph with $n$ vertices and $k$ cut
edges that maximizes the link residual closeness.

By Lemmas \ref{new1} and \ref{new2}, all cut edges of $G$ are pendant edges.
Let $V_0$ be the set of  vertices with degree one in $G$.
By Lemma \ref{AddE},  $G-V_0\cong K_{n-k}$.
Assume that $V(G-V_0)=\{v_1, \dots, v_{n-k}\}$ and that $d_G(v_i)=n-k-1+a_i$ for $i=1,\dots, n-k$,
with $a_1\ge \dots\ge a_{n-k}$. That is, $G$ is obtained from $K_{n-k}$  by attaching $a_i$ pendant edges at $v_i$ if $a_i>0$ for $i=1, \dots, n-k$.
By direct calculation,
\begin{align*}
C(G)&={n-k\choose 2}+\frac{1}{2}\sum_{i=1}^{n-k}{a_i\choose 2}\\
&\quad +\left(\frac{n-k-1}{2}+1\right)\sum_{i=1}^{n-k}a_i +\frac{1}{4}\sum_{1\le i<j\le n-k}a_ia_j\\
&=\frac{1}{2}(n-k)^2-\frac{1}{2}(n-k)+\frac{2(n-k)+1}{4}k\\
&\quad +\frac{1}{8}k^2 +\frac{1}{8}\sum_{i=1}^{n-k}a_i^2.
\end{align*}
If $a_i\ge 1$ with $i=1,\dots, n-k$, then for a pendant edge $e_i$ at  $v_i$,
\begin{align*}
&\quad C(G-e_i)\\
&=C(G)-\left(1+\frac{n-k-1}{2}+\frac{a_i-1}{2}+\frac{k-a_i}{4}\right)\\
&=C(G)-\frac{n}{2}+\frac{k}{4}-\frac{1}{4}a_i.
\end{align*}
For $1\le i<j\le n-k$ whether $a_ia_j=0$ or not,
\begin{align*}
&\quad C(G-v_iv_j)\\
&=C(G)-1+\frac{1}{2}-\frac{1}{2}(a_i+a_j)
 +\frac{1}{4}(a_i+a_j)\\
 &\quad  -\frac{1}{4}a_ia_j+\frac{1}{8}a_ia_j\\
&=C(G)-\frac{1}{2}-\frac{1}{4}(a_i+a_j)-\frac{1}{8}a_ia_j
\end{align*}
Recall that $a_1\ge\dots\ge a_{n-k}$,
so
$R^L(G)=\min \{C(G-e_1),C(G-v_1v_2) \}$, which is equal to the minimum of
$C(G)-\frac{n}{2}+\frac{k}{4}-\frac{1}{4}a_1$ and
$C(G)-\frac{1}{2}-\frac{1}{4}(a_1+a_2)-\frac{1}{8}a_1a_2$.

If $k=2$, then we have either $a_1=2$,  $a_2=\dots=a_{n-k}=0$ and so $G\cong C_{n,2}$ or $a_1=a_2=1$, $a_3=\dots =a_{n-k}=0$ and so $G\cong C_{n,2}'$.
In the former case,
\[
C(G-e_1)-C(G-v_1v_2)=1-\frac{1}{2}n<0,
\]
and in the latter case,
\[
C(G-e_1)-C(G-v_1v_2)=\frac{11}{8}-\frac{1}{2}n<0.
\]
Note that $C_{n,2}-e_1\cong C_{n,2}'-e_1$.
So $G\cong C_{n,2}, C_{n,2}'$, and   $R^L(G)=C(G-e_1)=\frac{1}{2}(n-k)^2+\frac{1}{2}$ in either case.

Suppose that $k\ge 3$. We claim that $a_2=0$.
Suppose to the contrary that $a_2\ge 1$.
Let $G'$ be the graph obtained from $G-V_0$ by attaching $k$ pendant edges at $v_1$.
Note that
\begin{align*}
&\quad C(G'-e_1)-C(G-e_1)\\
&=\frac{1}{4}(a_1-1)\sum_{i=2}^{n-k}a_i+\frac{1}{4}\sum_{2\le i<j\le n-k}a_ia_j\\
&\ge \frac{1}{4}(a_1-1)\sum_{i=2}^{n-k}a_i+\frac{1}{4}a_2\sum_{i=3}^{n-k}a_i\\
&> 0
\end{align*}
and
\begin{align*}
&\quad C(G'-v_1v_2)-C(G-v_1v_2)\\
& =\frac{1}{4}\sum_{1\le i<j\le n-k}a_ia_j-\frac{1}{4}\sum_{i=3}^{n-k}a_i+\frac{1}{8}a_1a_2 \\
& >\frac{1}{4}(a_1-1)\sum_{i=3}^{n-k}a_i\\
&\ge 0.
\end{align*}
So
\begin{align*}
R^L(G') & =\min\{C(G'-e_1),C(G'-v_1v_2) \}\\
&>\min\{C(G-e_1),C(G-v_1v_2) \}=R^L(G),
\end{align*}
a contradiction.
It follows that $a_2=0$,   so $G\cong C_{n,k}$ and
$
R^L(G)=C(G-e_1)=\frac{1}{2}(n-k)^2-\frac{1}{2}(n-k)+\frac{2(n-k)+1}{4}k+\frac{1}{4}k^2-\frac{n}{2}=\frac{n^2-nk}{2}-n+\frac{k^2+3k}{4}$,
as desired.
\end{proof}

From Theorem \ref{cutedges}, one immediately has

\begin{corollary}
Let $G$ be an $n$-vertex  connected graph with $k$ pendant
edges, where $1\le k\le n-3$. Then
\[
R^L(G)\le\frac{n^2-nk}{2}-n+\frac{k^2+3k}{4}
\]
with equality if and only if $G\cong C_{n,k}$ for $k\ne 2$ and  $G\cong C_{n,k}, C_{n,k}'$ for $k=2$.
\end{corollary}

\section{Number of cut vertices}

A vertex $v$ in a connected graph $G$ is a cut vertex if $G-v$ is disconnected.
For a connected graph $G$ on $n\ge 2$ vertices, it has at least two vertices that are not cut vertices, so $G$ possesses at most $n-2$ cut vertices. Moreover, if $G$ has $n-2$ cut vertices, then it is the $n$-vertex path.
If $G$ is a graph of order $n\ge 4$ with  one cut vertex, then, by Theorem \ref{conn},
\[
R^L(G)\le
\begin{cases}
\frac{n^2-3n+2}{2} & \mbox{if $n=4,n\ge 9$}\\
\frac{2n^2-7n+13}{4} & \mbox{if $5\le n\le 8$}
\end{cases}
\]
with equality if and only if $G\cong K_1\vee (K_2\cup K_{n-3})$ for $5\le n\le 8$, $G\cong K_1\vee (K_2\cup K_{n-3}),K_1\vee (K_1\cup K_{n-2})$ for $n=4,9$ and $G\cong K_1\vee (K_1\cup K_{n-2})$ for  $n\ge 10$.
In the remainder of this section, we consider the $n$-vertex connected graphs with $k$  cut vertices, where $2\le k\le n-3$.

For integers $r\ge 3$ and $a_1, \dots, a_r\ge 0$, we denote by $K^{a_1,\dots,a_r}$ the graph obtained from the complete graph
$K_r$ with vertex set $\{v_1,\dots,v_r\}$ by attaching a pendant path of length $a_i$ at $v_i$ if $a_i>0$ for $i=1,\dots,r$.

\begin{lemma}\label{vertexc}
For fixed integers $r\ge 3$ and $s\ge 2$, let
\[
\mathcal{G}(r,s)=\left\{K^{a_1,\dots,a_r}: a_1\ge \dots\ge a_r\ge 0, \sum_{i=1}^ra_i=s\right\}.
\]
Then
$G=K^{a_1,\dots,a_r}$ maximizes the link residual closeness in $\mathcal{G}(r,s)$ if and only if $a_1-a_r=0,1$.
\end{lemma}

\begin{proof}
Let $G=K^{a_1,\dots,a_r}\in \mathcal{G}(r,s)$. First, we show that
\begin{equation}\label{MM}
\begin{aligned}
R^L(G)&=2s-4+2^{2-a_1}+\sum_{i=2}^{r}2^{-a_i}\\
&\quad +\sum_{2\le i<j\le r}(2-2^{-a_i})(2-2^{-a_j}).
\end{aligned}
\end{equation}

For $i=1,\dots,r$, let $Q_i$ be the component of $G-E(K_r)$ containing $v_i$.
Then $Q_i$ is a path $P_{a_i+1}$ with one terminal vertex $v_i$, say  $Q_i=v_{i,0}v_{i,1}\dots v_{i,a_i}$ with $v_{i,0}=v_i$.
Let $V_i=V(Q_i)$ for $i=1,\dots, r$.

Let $u\in V_i$ and $v\in V_j$ with $1\le i<j\le r$.
Note that $d_{G-v_iv_j}(u,v)=d_G(u,v)+1$, and if $\{w,z\}\ne\{u,v\}$, then $d_{G-v_iv_j}(w,z)=d_G(w,z)$. So
\begin{align*}
&\quad C(G-v_iv_j)\\
&=C(G)-2\sum_{u\in V_i}\sum_{v\in V_j} \left(2^{-d_G(u,v)}-2^{-d_{G-v_iv_j}(u,v)}\right)\\
&=C(G)-\sum_{u\in V_i}\sum_{v\in V_j} 2^{-d_G(u,v)}\\
&=C(G)-\frac{1}{2}\sum_{u\in V_i}2^{-d_G(u,v_i)}\sum_{v\in V_j} 2^{-d_G(v_j,v)}\\
&=C(G)-\frac{1}{2}\sum_{s=0}^{a_i} 2^{-s} \sum_{s=0}^{a_j}2^{-s}\\
&=C(G)-\frac{1}{2}(2-2^{-a_i})(2-2^{-a_j}).
\end{align*}
As $a_1\ge\dots\ge a_r$,
we see that for $1\le i<j\le r$,
\begin{equation}\label{eq1}
\begin{aligned}
C(G-v_iv_j) & \ge C(G-v_1v_2)\\
&=C(G)-\frac{1}{2}(2-2^{-a_1})(2-2^{-a_2}).
\end{aligned}
\end{equation}
Note that
\begin{align*}
&\quad C(G-v_iv_{i,1})\\
& = C(G)-2\sum_{u\in V_i\setminus \{v_i\}}\sum_{v\in  V(G)\setminus V_i}2^{-d_G(u,v)}\\
&\quad -2\sum_{w\in V_i\setminus\{v_i\}} 2^{-d_G(v_i,w)}\\
& =C(G)-\sum_{u\in V_i\setminus \{v_i\}}2^{-d_G(u,v_i)}\sum_{\substack{1\le j\le r// j\ne i}}\sum_{v\in V_j}2^{-d_G(v_j,v)}\\
&\quad -\sum_{j=1}^{a_i}2^{-j+1}\\
&=C(G)-(1-2^{-a_i})\sum_{\substack{1\le j\le r// j\ne i}} (2-2^{-a_j})-\left(2-2^{-(a_i-1)}\right)\\
&=C(G)-(1-2^{-a_i})\left(\sum_{j=1}^r (2-2^{-a_j})+2^{-a_i}\right)
\end{align*}
and
$(1-2^{-a_i})\left(\sum_{j=1}^r (2-2^{-a_j})+2^{-a_i}\right)$
is maximized if $i=1$ because $a_1\ge \dots\ge a_r$.
Thus, by Lemma \ref{pendant}, for $i=1,\dots,r$ with $a_i>0$ and $j=1,\dots,a_i$, we get
\begin{equation}\label{eq2}
\begin{aligned}
 &\quad C(G-v_{i,j-1}v_{i,j})\\
 &\ge C(G-v_iv_{i,1}) \ge C(G-v_1v_{1,1})\\
&=C(G)-(1-2^{-a_1})\sum_{j=2}^r(2-2^{-a_j}) \\
&\quad -\left(2-2^{-(a_1-1)}\right).
\end{aligned}
\end{equation}
Now it follows from Eqs. \eqref{eq1} and \eqref{eq2} that
\begin{align*}
& \quad  C(G-v_1v_{1,1})-C(G-v_1v_2)\\
&=\frac{1}{2}(2-2^{-a_1})(2-2^{-a_2})-(1-2^{-a_1})\sum_{j=2}^r(2-2^{-a_j})\\
&\quad -\left(2-2^{-(a_1-1)}\right)\\
&<\frac{1}{2}(2-2^{-a_1})(2-2^{-a_2})-(1-2^{-a_1})(2-2^{-a_2}) \\
&\quad -\left(2-2^{-(a_1-1)}\right)\\
&=-\frac{1}{2}(2-2^{-a_1})(2-2^{-a_2})+2^{-(a_1-1)}-2^{-a_2}\\
&\le \begin{cases}
-\frac{1}{2}(2-2^{-a_1})(2-2^{-a_2}) & \mbox{if $a_1\ge a_2+1$}\\[1mm]
-2+3\cdot 2^{-a_1}-2^{-2a_1-1} & \mbox{if $a_1= a_2$}
\end{cases}\\
&< 0.
\end{align*}
Therefore $R^L(G)=C(G-v_1v_{1,1})$.
Note that $G[V_i\cup \{v_1\}]\cong P_{a_i+2}$ for $i=2,\dots, r$. By direct calculation, we find that
\begin{align*}
&\quad C(G-v_1v_{1,1})-C(P_{a_1})-\sum_{i=2 }^rC(P_{a_i+2})\\
&=\sum_{2\le i<j\le r}\sum_{u\in V_i}2^{-d_G(u,v_i)}\sum_{v\in V_j}2^{-d_G(v_j,v)}\\
&=\sum_{2\le i<j\le r}\sum_{s=0}^{a_i}2^{-s}\sum_{t=0}^{a_j}2^{-t}\\
&=\sum_{2\le i<j\le r}(2-2^{-a_i})(2-2^{-a_j}),
\end{align*}
from which  \eqref{MM} follows, as $C(P_{a_1})+\sum_{i=2 }^rC(P_{a_i+2})=2a_1-4+2^{2-a_1}+\sum_{i=2}^r \left(2(a_i+2)-4+2^{-a_i}\right)$.

Denote the expression for $R^L(G)$ in \eqref{MM} by $f(a_1,\dots, a_r)$.
Suppose that $a_1\ge a_r+2$.
Let $p$ be the largest integer  and $q$ be the smallest integer such that $a_p=a_1$ and $a_p\ge a_q+2$. Then $1\le p<q\le r$. For $i=1, \dots, r$, let $b_p=a_p-1$, $b_q=a_q+1$ and $b_i=a_i$ if $i\ne p,q$.
If  $p=1$, then
\begin{align*}
&\quad f(b_1,\dots, b_r)-f(a_1, \dots, a_r)\\
&=2^{2-b_1}+\sum_{i=2}^{r}2^{-b_i}+\sum_{2\le i<j\le r}(2-2^{-b_i})(2-2^{-b_j})\\
&\quad  -2^{2-a_1}-\sum_{i=2}^{r}2^{-a_i}-\sum_{2\le i<j\le r}(2-2^{-a_i})(2-2^{-a_j})\\
&=2^{2-a_1}-2^{-1-a_q}+2^{-1-a_q}\sum_{\substack{2\le i\le r// i\ne q}}(2-2^{-a_i})\\
&>0,
\end{align*}
and if  $p>1$, then
\begin{align*}
&\quad f(b_1,\dots, b_r)-f(a_1, \dots, a_r)\\
&=2^{2-b_1}+\sum_{i=2}^{r}2^{-b_i}+\sum_{2\le i<j\le r}(2-2^{-b_i})(2-2^{-b_j})\\
&\quad -2^{2-a_1}-\sum_{i=2}^{r}2^{-a_i} -\sum_{2\le i<j\le r}(2-2^{-a_i})(2-2^{-a_j})\\
&=(2^{-1-a_q}-2^{-a_p})\left(\sum_{\substack{1\le i\le r// i\neq p,q}}(2-2^{-a_i})+1\right)\\
&>0.
\end{align*}
So $f(a_1,\dots, a_r)$ is maximized if and only if
 $a_1-a_r=0,1$.
\end{proof}

Denote by $PK_{r+s, s}$ the graph $K^{a_1,\dots,a_r}$  with $a_1\ge \dots\ge a_r$, $\sum_{i=1}^ra_i=s$, and $a_1-a_r=0,1$.

A clique of a graph $G$ is an induced subgraph of $G$ that is complete. A block is trivial if it has at exactly two vertices, and it is nontrivial if it has at least three vertices.

For integers $b_1$, $b_2$ and $n$ with $b_2\ge b_1\ge 3$ and $2(b_1+b_2)-3=n$, let $H_n(b_1, b_2)$ be the $n$-vertex graph consisting of two cliques of size $b_1$ and $b_2$ with a common vertex  and  a pendant edge at each other vertex.

\begin{theorem}\label{cutvertex}
Let $G$ be an $n$-vertex connected graph with $k$ cut vertices, where $2\le k\le n-3$.
Let  $k=(n-k)q+r$ with $0\le r<k$. If
$(n,k)= (9,5)$, then $R^L(G)\le 16$ with equality if and if $G\cong H_9(3,3)$.
If $(n,k)\ne (9,5)$, then
\begin{align*}
R^L(G)&\le 2(n^2+k^2-2nk-3n+4k)\\
&\quad -2^{-q}(2n^2+2k^2-4nk-7n+7k-1)\\
&\quad +2^{-2q-1}(n^2+k^2-2nk-3n+3k+2)
\end{align*}
if $r=0$, and
\begin{align*}
R^L(G)&\le (n-k)(2n-2k-6)+2k\\
&\quad -2^{-q}\left((n-k)(2n-2k-r-6)+\frac{5r+1}{2} \right)\\
&\quad +2^{-2q-2}\left(2(n-k-2)(n-k-r)+\frac{r^2-3r+2}{2}  \right)
\end{align*}
otherwise,
and the bound for $R^L(G)$ is attained
if and only if $G\cong PK_{n,k}$ when $(n,k)\ne (11,6)$ and $G\cong PK_{11,6}, H_{11}(3,4)$ when $(n,k)= (11,6)$.
\end{theorem}

\begin{proof}
Suppose that $G$ is an $n$-vertex connected graph with $k$ cut vertices that maximizes the link residual closeness.

By Lemma \ref{AddE}, all blocks of $G$ are cliques with at least two vertices and each cut
vertex of $G$ is contained in exactly two blocks. So there are
exactly $k+1$ blocks in $G$.
As $k\le n-3$, there is at least one nontrivial block of $G$.
We call a nontrivial block $B$  of $G$ a pendant block if for some vertex $x$ in this block, the component of $G-E(B)$ containing $x$ is a path (that may be trivial) with one end vertex being $x$. In this case, we call this path
the path at $x$.
We choose  a pendant block $B_1$ so that for one of its vertex, say $u_0$,   the length $\ell$ of the  path at $u_0$ is minimum among all paths at vertices of all pendant blocks of $G$.

We will prove that either there is exactly one nontrivial block of $G$ or  $G\cong H_{9}(3,3), H_{11}(3,4)$. To this end,  we suppose that
there are at least two nontrivial blocks of $G$. Then it need only to show that
 $G\cong H_{9}(3,3), H_{11}(3,4)$.

Let $W$ be the set of vertices in all nontrivial blocks different from $B_1$. Let $t-1=\min\{d_G(u,u'): u\in V(B_1), u'\in W\}$. Assume that
$t-1=d_G(u_1,u_t)$ with $u_t\in V(B_2)$ for some nontrivial block $B_2$ different from $B_1$. Then there is a unique path, say  $u_1\dots u_t$.
For any $w\in V(B_1)$ ($w\in V(B_2)$, respectively), let $T_w$ be the   the component of $G-E(B_1)$ ($G-E(B_2)$, respectively) containing $w$.
Let $N=V(B_1)\setminus \{u_0,u_1\}$.
Let $V_0=V(T_{u_0})$,
$V_1=\cup_{w\in N}V(T_w)$
and $V_2=\cup_{w\in V(B_2)\setminus \{u_t\}}V(T_w)$.
Let \[
G'=G-\{yu_0,yu_1:y\in N\}+\{yw:y\in N, w\in V(B_2)\}.
\]
It is obvious that $G'$ is an $n$-vertex connected graph with $k$ cut vertices.

Assume that $R^L(G')=C(G'-e)$ with $e\in E(G')$.
By Lemma \ref{pendant}, $e\not\in E(G[V_0])\cup\{u_{i-1}u_i:1\le i\le t-1\}$. Let $b_i=|V(B_i)|$ for $i=1,2$.

\noindent
{\bf Claim 1.}  $e\notin E(G[V_1])\cup E(G[V_2])$  and $e\ne u_tx$ for any $x\in V(B_2)\setminus\{u_t\}$.

Suppose that this is not true. That is, $e\in E(G[V_1])\cup E(G[V_2])$  or $e= u_tx$ for some $x\in V(B_2)\setminus\{u_t\}$.
As we pass from $G-e$ to $G'-e$, the distance between any pair of vertices in $V_0\cup V_2\cup \{u_1,\dots,u_t\}$ and in $V_1$ remains unchanged.
Moreover, for any $z\in V_1$,
\[
\sum_{i=1}^{t}\left(2^{-d_{G'-e}(z,u_i)}-2^{-d_{G-e}(z,u_i)}\right)=0.
\]
So
\begin{equation}\label{U1}
R^L(G')-R^L(G)\ge C(G'-e)-C(G-e)=2\sum_{z\in V_1}F(z)
\end{equation}
with
\[
F(z)=\sum_{w\in V_0\cup V_2}\left(2^{-d_{G'-e}(z,w)}-2^{-d_{G-e}(z,w)}\right).
\]
We will show $F(z)\ge 0$ for all $z\in V_1$, and $F(z)>0$ for some $z\in N\subseteq V_1$, so from \eqref{U1}, we have
$R^L(G')>R^L(G)$, a contradiction.

Suppose first that $e\in E(G[V_1])\cup E(G[V_2])$.
For any $z\in V_1$, we have
\[
d_{G'-e}(z,w)=d_{G-e}(z,w)
\begin{cases}
+t & \mbox{if $w\in V_0$},\\
-t & \mbox{if $w\in V_2$}.
\end{cases}
\]
Note also that
\begin{align*}
&\quad \sum_{w\in V_2}2^{-d_{G-e}(u_t,w)}-\sum_{w\in V_0}2^{-d_{G-e}(u_0,w)-1}\\
&=\sum_{w\in V_2}2^{-d_{G-e}(u_t,w)}-\sum_{i=1}^{\ell+1}2^{-i}\\
&\ge (b_2-1)-1+2^{-\ell-1}\\
&>0.
\end{align*}
So for $z\in V_1$,
\begin{align*}
 F(z)
&=(2^{-t}-1)\sum_{w\in V_0}2^{-d_{G-e}(z,w)}\\
&\quad +(2^{t}-1)\sum_{w\in V_2}2^{-d_{G-e}(z,w)}\\
&=(2-2^{-t+1})2^{-d_{G-e}(z,u_0)}\\
&\quad \cdot \left(\sum_{w\in V_2}2^{-d_{G-e}(u_t,w)}-\sum_{w\in V_0}2^{-d_{G-e}(u_0,w)-1}\right)\\
&\ge 0
\end{align*}
and the inequality is strict if $d_{G-e}(z,u_0)<\infty$, for example, for $z\in N$.

Suppose next that $e=u_tx$ for some $x\in V(B_2)\setminus \{u_t\}$.
For any $z\in V_1$,
\[
d_{G'-e}(z,w)=d_{G-e}(z,w)\begin{cases}
+t & \mbox{if $w\in V_0$},\\
-t & \mbox{if $w\in V_2\setminus V(T_x)$},\\
-t-1 & \mbox{if $w\in V(T_x)$},
\end{cases}
\]
so
\begin{align*}
 F(z)
&=(2^{-t}-1)\sum_{w\in V_0}2^{-d_{G-e}(z,w)}\\
&\quad +(2^{t}-1)\sum_{w\in V_2\setminus V(T_x)}2^{-d_{G-e}(z,w)}\\
&\quad +(2^{t+1}-1)\sum_{w\in V(T_x)}2^{-d_{G-e}(z,w)}\\
&=(2-2^{-t+1})2^{-d_{G}(z,u_0)}\\
&\quad \cdot \left(\sum_{w\in V_2\setminus V(T_x)}2^{-d_{G}(u_t,w)}-\sum_{w\in V_0}2^{-d_{G}(u_0,w)-1}\right)\\
&\quad +(2-2^{-t})2^{-d_{G}(z,u_0)}\sum_{w\in V(T_x)}2^{-d_{G}(u_t,w)}\\
&\ge (2-2^{-t+1})2^{-d_{G}(z,u_0)}\\
&\quad \cdot\left(\sum_{w\in V_2}2^{-d_{G}(u_t,w)}-\sum_{i=1}^{\ell+1}2^{-i} \right)\\
&>0,
\end{align*}
as desired.
This proves Claim 1.

\noindent
{\bf Claim 2.}  $e\ne uv$ for any $u\in N$ and any $v\in V(B_2)$.

Suppose that this is not true. That is, $e= uv$ for some $u\in N$ and some $v\in V(B_2)$.
Let $f=uu_1$.
As we pass from $G-f$ to $G'-e$, the distance between any pair of vertices in $V_0\cup V_2\cup \{u_1,\dots,u_t\}$ and in $V_1$ remains unchanged.
So
\[
R^L(G')-R^L(G)\ge C(G'-e)-C(G-f)=2\sum_{z\in V_1}F(z)
\]
with
\[
F(z)=\sum_{w\in V_0\cup V_2\cup \{u_1,\dots,u_t\}}\left(2^{-d_{G'-e}(z,w)}-2^{-d_{G-f}(z,w)}\right).
\]
We will show that $F(z)>0$ for $z\in V_1$, so $R^L(G')>R^L(G)$, a contradiction.

Suppose first that $v=u_t$, that is, $e=uu_t$.
For any $z\in V_1$,
\[
\sum_{i=1}^t\left(2^{-d_{G'-e}(z,u_i)}-2^{-d_{G-f}(z,u_i)}\right)=0,
\]
so
\[
F(z)=\sum_{w\in V_0\cup V_2}\left(2^{-d_{G'-e}(z,w)}-2^{-d_{G-f}(z,w)}\right).
\]
If $z\in V(T_u)$, then
\[
d_{G'-e}(z,w)=d_{G-f}(z,w)\begin{cases}
+t+1&\mbox{if $w\in V_0$},\\
-t-1&\mbox{if $w\in V_2$},
\end{cases}
\]
so
\begin{align*}
F(z)
&=(2-2^{-t})2^{-d_{G}(z,u_0)}\\
&\quad \cdot \left(\sum_{w\in V_2}2^{-d_G(u_t,w)}-\sum_{w\in V_0}2^{-d_G(u_0,w)-1} \right)\\
&>0.
\end{align*}
If $z\in V_1\setminus V(T_u)$, then
\[
d_{G'-e}(z,w)=d_{G-f}(z,w)\begin{cases}
+t&\mbox{if $w\in V_0$},\\
-t&\mbox{if $w\in V_2$},
\end{cases}
\]
so
\begin{align*}
F(z)
&=(2-2^{-t+1})2^{-d_{G}(z,u_0)}\\
&\quad \cdot \left(\sum_{w\in V_2}2^{-d_G(u_t,w)}-\sum_{w\in V_0}2^{-d_G(u_0,w)-1} \right)\\
&>0.
\end{align*}

Suppose next that $v\in V(B_2)\setminus \{u_t\}$.
If $z\in V(T_u)$, then
\[
\sum_{i=0}^t(2^{-d_{G'-e}(z,u_i)}-2^{-d_{G-f}(z,u_i)})=0
\]
and
\[
d_{G'-e}(z,w)=d_{G-f}(z,w)\begin{cases}
+t&\mbox{if $w\in V_0\setminus\{u_0\}$},\\
-t&\mbox{if $w\in V(T_v)$},\\
-t-1&\mbox{if $w\in V_2\setminus V(T_v)$},
\end{cases}
\]
so
\begin{align*}
F(z)&=\sum_{w\in V_0\cup V_2\setminus\{u_0\}}\left(2^{-d_{G'-e}(z,w)}-2^{-d_{G-f}(z,w)} \right)\\
&=(1-2^{-t})2^{-d_G(z,u_0)}\\
&\quad \cdot \left(\sum_{w\in V(T_v)}2^{-d_G(u_t,w)}-\sum_{w\in V_0\setminus\{u_0\}}2^{-d_G(u_0,w)} \right)\\
&\quad +(2-2^{-t})2^{-d_G(z,u_0)}\sum_{w\in V_2\setminus V(T_v)}2^{-d_G(u_t,w)}\\
&\ge (1-2^{-t})2^{-d_G(z,u_0)}\\
&\quad \cdot \left(\sum_{w\in V_2}2^{-d_G(u_t,w)}-\sum_{w\in V_0\setminus\{u_0\}}2^{-d_G(u_0,w)} \right)\\
&>0.
\end{align*}
If $z\in V_1\setminus V(T_u)$, then
\[
\sum_{i=1}^t(2^{-d_{G'-e}(z,u_i)}-2^{-d_{G-f}(z,u_i)})=0
\]
and
\[
d_{G'-e}(z,w)=d_{G-f}(z,w)\begin{cases}
+t&\mbox{if $w\in V_0\setminus\{u_0\}$},\\
-t&\mbox{if $w\in V_2$},
\end{cases}
\]
so
\begin{align*}
F(z)&=\sum_{w\in V_0\cup V_2}\left(2^{-d_{G'-e}(z,w)}-2^{-d_{G-f}(z,w)} \right)\\
&=(2-2^{-t+1})2^{-d_G(z,u_0)}\\
&\quad \cdot \left(\sum_{w\in V_2}2^{-d_G(u_t,w)}-\sum_{w\in V_0\setminus \{u_0\}}2^{-d_G(u_0,w)-1} \right)\\
&>0.
\end{align*}

In any case, we have $F(z)>0$ for $z\in V_1$. This proves Claim 2.

By Claims 1 and 2, we have   $e=u_{t-1}u_t$.

From the definition of $\ell$, one has $\sum_{w\in V_0}2^{-d_G(u_1,w)}=\sum_{i=1}^{\ell+1}2^{-i}=1-2^{-\ell-1}$.
Then
\begin{equation*}
\sum_{w\in V(T_z)}2^{-d_G(u_1,w)}\ge \sum_{w\in V_0}2^{-d_G(u_1,w)}=1-2^{-\ell-1}
\end{equation*}
for any $z\in N$ whether $T_z$ is a pendant path at $z$ or not,
so
\begin{equation}\label{V1}
\begin{aligned}
\sum_{w\in V_1}2^{-d_G(w,u_0)}&=\sum_{w\in V_1}2^{-d_G(w,u_1)}\\
&=\sum_{z\in N}\sum_{w\in V(T_z)}2^{-d_G(u_1,w)}\\
&\ge (b_1-2)(1-2^{-\ell-1}).
\end{aligned}
\end{equation}
Similarly,
\begin{equation*}
\sum_{w\in V(T_z)}2^{-d_G(u_t,w)}\ge \sum_{w\in V_0}2^{-d_G(u_1,w)}=1-2^{-\ell-1}
\end{equation*}
for any $z\in V(B_2)\setminus\{u_t\}$, so
\begin{equation}\label{V2}
\begin{aligned}
\sum_{w\in V_2}2^{-d_G(u_t,w)}&=\sum_{z\in V(B_2)\setminus\{u_t\}}\sum_{w\in V(T_z)}2^{-d_G(u_1,w)}\\
&\ge (b_2-1)(1-2^{-\ell-1}).
\end{aligned}
\end{equation}
\noindent
{\bf Claim 3.} $t=1$.

Suppose that $t\ge 2$.
As we pass from $G-e$ to $G'-e$, the distance between any pair of vertices in $V_0\cup \{u_1,\dots,u_{t-1}\}$, in $V_1$ and in $V_2\cup \{u_t\}$ remains unchanged.
Then
\[
C(G'-e)-C(G-e)=2\sum_{z\in V_1}F(z)
\]
with
\begin{align*}
F(z)&=\sum_{w\in V_2\cup \{u_t\}}2^{-d_{G'}(z,w)}\\
&\quad -\sum_{w\in V_0\cup \{u_1,\dots,u_{t-1}\}}2^{-d_G(z,w)}.
\end{align*}
For any $z\in V_1$, there is some vertex $x\in N$ such that $z\in V(T_x)$.
Then
\[
d_{G'}(z,w)=d_{G}(z,x)+d_{G'}(x,w)
\]
for any $w\in V_2$
and
\[
d_{G}(z,w)=d_G(z,x)+d_G(x,w)
\]
for any $w\in V_0\cup \{u_1,\dots,u_{t-1}\}$.
So
\begin{align*}
&\quad  F(z)\cdot 2^{d_{G}(z,x)}\\
&=\sum_{w\in V_2\cup \{u_t\}}2^{-d_{G'}(x,w)}-\sum_{w\in V_0\cup \{u_1,\dots,u_{t-1}\}}2^{-d_G(x,w)}\\
&=\sum_{w\in V_2\cup \{u_t\}}2^{-d_{G'}(x,w)}-\sum_{i=1}^{\ell+1}2^{-i}-\sum_{i=1}^{t}2^{-i}\\
&\ge 1+(b_2-1)(1-2^{-\ell-1})-(1-2^{-\ell-1})-(1-2^{-t})\\
&\ge 1+1-2^{-\ell-1}-1+2{-t}\\
&>0,
\end{align*}
where the first inequality follows as $\sum_{w\in V_2}2^{-d_{G'}(x,w)}=\sum_{w\in V_2}2^{-d_{G'}(u_t,w)}\ge (b_2-1)(1-2^{-\ell-1})$ by Eq. \eqref{V2}.
It thus follows that $R^L(G')=C(G'-e)>C(G-e)\ge R^L(G)$,
a contradiction. So Claim 3 follows.

By Claim 3, $t=1$. That is,  $V(B_1)\cap V(B_2)=\{u_1\}$ and  $e=u_0u_1$.

\noindent
{\bf Claim 4.}  For any $y\in N\cup V(B_2)\setminus\{u_1\}$,
$y$ lies in some nontrivial block different from $B_1$ and $B_2$ or $T_y$ is a pendant path of length $\ell$ at $y$.

Suppose to the contrary that $y$ lies outside any nontrivial block different from $B_1$ and $B_2$ and
$T_y$ is not a pendant path of length of length $\ell$.
Then we have two cases: (a) $T_y$ is a pendant path of length at least $\ell+1$, or (b) $T_y$ is not a path, which implies that there is a nontrivial block not containing $y$ in $T_y$.
In either case, $y$ has a unique neighbor $y'$ in $T_y$,
$\sum_{w\in V(T_y) \setminus \{y\}}2^{-d_{G}(y,w)}-\sum_{w\in V_0}2^{-d_{G}(u_1,w)}\ge 0$, and so
\begin{align*}
&\quad C(G'-e)- C(G'-yy')\\
&\ge 2\sum_{z\in V(G)\setminus (V_0\cup V(T_y) \setminus \{x\})}\left(\sum_{w\in V(T_y) \setminus \{y\}}2^{-d_{G'}(z,w)}-\sum_{w\in V_0}2^{-d_{G'}(z,w)}\right)\\
&=2\sum_{z\in V(G)\setminus (V_0\cup V(T_y) \setminus \{y\})}2^{-d_{G}(z,u_1)}\left(\sum_{w\in V(T_y) \setminus \{y\}}2^{-d_{G}(y,w)}-\sum_{w\in V_0}2^{-d_{G}(u_1,w)}\right)\\
&\ge 0.
\end{align*}
On the other hand,
$C(G'-yy')-C(G-yy')>0$ as in the argument in Claim 1.
So
$R^L(G) \le C(G-yy')<C(G'-yy') \le C(G'-e)=R^L(G')$,
a  contradiction.
This proves Claim 4.

\noindent
{\bf Claim 5.} $\ell\ge 1$.

Suppose that $\ell=0$.
Since $G$ possesses  $k\ge 2$ cut vertices, there is some cut vertex $x\in (V(B_1)\cup V(B_2))\setminus \{u_0, u_1\}$. By Claim 4, $x$ lies in some nontrivial block $B$ of $G$ different from $B_1$ and $B_2$.

Assume that $x\in V(B_1)$ as the case when $x\in V(B_2)$ may be proved similarly. Assume further that $B$ is such a nontrivial block with the largest size $b$.
Note that
\begin{equation}\label{LY}
\sum_{w\in V(T_x)}2^{-d_{G}(u_0,w)} \ge 2^{-1}+(b-1)2^{-2}.
\end{equation}
Let  $f=u_1x$.
As we pass from $G-f$ to $G'-e$, the distance between any pair of vertices in $V_1$, in $V_1\cup \{u_1\}\setminus V(T_x)$ and in $V_2\cup\{u_1\}$ remains unchanged.
So
\begin{align*}
&\quad C(G'-e)-C(G-f)\\
&=2\sum_{z\in V_2}\sum_{w\in V_1}\left(2^{-d_{G'-e}(z,w)}-2^{-d_{G-f}(z,w)} \right) \\
&\quad +2\sum_{w\in V(T_x)}\left(2^{-d_{G'-e}(u_1,w)}-2^{-d_{G-f}(u_1,w)} \right)\\
&\quad -2\sum_{w\in V_1\cup V_2\cup \{u_1\}}2^{-d_{G-f}(u_0,w)}.
\end{align*}
For $z\in V_2$,
\[
d_{G'-e}(z,w)=d_{G-f}(z,w)\begin{cases}
-1&\mbox{if $w\in V_1\setminus V(T_x)$},\\
-2&\mbox{if $w\in V(T_x)$},
\end{cases}
\]
so
\begin{align*}
&\quad \sum_{w\in V_1}\left( 2^{-d_{G'-e}(z,w)}-2^{-d_{G-f}(z,w)} \right)\\
&=\sum_{w\in V_1\setminus V(T_x)}2^{-d_{G-f}(z,w)} +3\sum_{w\in V(T_x)}2^{-d_{G-f}(z,w)}\\
&=2^{-d_{G}(z,u_1)}\sum_{w\in V_1\setminus V(T_x)} 2^{-d_{G}(u_0,w)} \\
&\quad +2^{-d_{G}(z,u_1)-1}\cdot3\sum_{w\in V(T_x)}2^{-d_{G}(u_0,w)}.
\end{align*}
For any $w\in V(T_x)$,
$d_{G'-e}(u_1,w)=d_{G-f}(u_1,w)-1$,
so
\[
2^{-d_{G'-e}(u_1,w)}-2^{-d_{G-f}(u_1,w)}=2^{-d_G(u_0,w)-1}.
\]
Note that
\begin{align*}
&\quad \sum_{w\in V_1\cup  V_2\cup \{u_1\}}2^{-d_{G-f}(u_0,w)}\\
&=\sum_{w\in V_1\setminus V(T_x)}2^{-d_G(u_0,w)} +\sum_{w\in V(T_x)}2^{-d_G(u_0,w)}\\
 &\quad +\sum_{z\in V_2}+2^{-d_G(z,u_1)-1}+2^{-1}.
\end{align*}
Thus
\begin{align*}
&\quad  C(G'-e)-C(G-f)\\
&=2\left(\sum_{z\in V_2}2^{-d_{G}(z,u_1)} -1\right)\sum_{w\in V_1\setminus V(T_x)} 2^{-d_{G}(u_0,w)}\\
&\quad +2\left(3\sum_{z\in V_2} 2^{-d_{G}(z,u_1)-1}-2^{-1} \right)\\
&\quad \cdot\left(\sum_{w\in V(T_x)}2^{-d_{G}(u_0,w)}-3^{-1}\right)-\frac{4}{3}.
\end{align*}
If $b_2\ge 4$ or $b\ge 4$, then by Eqs.~\eqref{V1},~\eqref{V2} and~\eqref{LY},
\begin{align*}
&\quad  C(G'-e)-C(G-f)\\
&\ge 2\left( (b_2-1)2^{-1}-1\right)(b_1-3)2^{-1}\\
&\quad + 2\left(3(b_2-1)2^{-2} -2^{-1} \right)\left(2^{-1}+(b-1)2^{-2}-3^{-1}\right)-\frac{4}{3}\\
&>0,
\end{align*}
and hence $R^L(G')=C(G'-e)>C(G-f)\ge R^L(G)$,
a contradiction.
So $b_2=b=3$.
If there is a nontrivial block $B'$ different  from $B_1, B_2, B$ such that $B'$ and  $B_2$ share a common vertex,
then
\begin{align*}
&\quad  C(G'-e)-C(G-f)\\
&> 2\left( (b_2-1)2^{-1}-1\right)(b_1-3)2^{-1}\\
&\quad + 2\left(3(b_2-1)2^{-2} -2^{-1} \right)\left(2^{-1}+(b-1)2^{-2}-3^{-1}\right)-\frac{4}{3}\\
&=0,
\end{align*}
so $R^L(G')>R^L(G)$, also a contradiction. So
 there is no nontrivial block $B'$ different  from $B_1$ such that $B'$ and  $B_2$ share a common vertex. From this fact, Claim 4 and the assumption that $\ell=0$, we see that
each vertex in $V(B_2)\setminus\{u_1\}$ is not a cut vertex.

By similar argument,  each vertex in $V(B)\setminus\{x\}$ is not a cut vertex.

Let $V(B_2)=\{u_1,u_2,u_3\}$,
$V(B)=\{x,x_1,x_2\}$ and
\[
G^*=G-u_2u_3-x_1x_2+x_2u_2+\{x_2w,u_2w:w\in V(B_1)\}.
\]
By similar argument as in Claims 1 and 2, $R(G^*)=C(G^*-u_1u_3)=C(G^*-xx_1)$.
Let $f_1=u_1u_3$, $f_2=u_1x$ and $W_1=V(G)\setminus(V(B_2)\cup V(B))$. Then
\begin{align*}
&\quad C(G^*-f_1)-C(G-f_2)\\
&=2\sum_{w\in W_1\cup\{x,x_1\}}\left(2^{-d_{G^*-f_1}(u_2,w)}-2^{-d_{G-f_2}(u_2,w)} \right)\\
&\quad +2\sum_{w\in W_1\cup\{u_1,u_2,x_1\}}\left(2^{-d_{G^*-f_1}(x_2,w)}-2^{-d_{G-f_2}(x_2,w)} \right)\\
&\quad +2\sum_{w\in\{x,x_1\}}(2^{-d_{G^*-f_1}(u_1,w)}-2^{-d_{G-f_2}(u_1,w)})-2\sum_{w\in V(G)\setminus\{u_3\}}2^{-d_{G-f_2}(u_3,w)}\\
&=2\left(\sum_{w\in W_1}2^{-d_G(u_2,w)}+3\cdot 2^{-d_G(u_2,x)-1}+3\cdot 2^{-d_{G}(u_2,x_1)-1}\right)\\
&\quad +2\left(\sum_{w\in W_1}2^{-d_G(x_2,w)}+3\cdot 2^{-d_{G}(x_2,u_1)-1}+7\cdot 2^{-d_G(u_2,x_2)-1}\right)+2^{-d_G(u_1,x_1)}\\
&\quad -2\left(\sum_{w\in W_1}2^{-d_G(u_3,w)}+2\cdot2^{-d_G(u_3,u_2)}+2^{-d_G(u_3,x)-1}+2\cdot2^{-d_G(u_3,x_1)-1}\right)\\
&=2\sum_{w\in W_1}2^{-d_G(u_2,w)}+3\cdot 2^{-2}+3\cdot 2^{-3}+3\cdot 2^{-2}+7\cdot 2^{-3}+2^{-2}-2-2^{-2}-2^{-2}\\
&=2\sum_{w\in W_1}2^{-d_G(u_2,w)}+\frac{1}{2}\\
&>0,
\end{align*}
and hence
$R^L(G^*)=C(G^*-f_1)>C(G-f_2)\ge R^L(G)$, a contradiction.
 This proves Claim 5.

By Claim 5, $\ell\ge 1$.
Denote by $u_0'$ the neighbor of $u_0$ in $V_0$.
Let $e'=u_0u_0'$.
As we pass from $G-e'$ to $G'-e$, the distance between any pair of vertices in $V_0$, in $V_1$ and in $V_2$ remains unchanged. Note that
\[
d_{G'-e}(z,w)=d_{G-e'}(z,w)-1
\]
for any $z\in V_1$, $w\in V_2$.
Thus
\begin{align*}
&\quad C(G'-e)-C(G-e') \\
&=2\sum_{z\in V_1}\sum_{w\in V_2}\left(2^{-d_{G'-e}(z,w)}-2^{-d_{G-e'}(z,w)} \right)\\
&\quad +2\sum_{z\in  V_0\setminus\{u_0\}}2^{-d_{G'-e}(z,u_0)} -2\sum_{z\in V_1\cup V_2\cup \{u_1\}}2^{-d_{G-e'}(z,u_0)}\\
&=2\sum_{z\in V_1}2^{-d_{G}(z,u_1)}\sum_{w\in V_2}2^{-d_{G}(u_1,w)}+\sum_{i=0}^{\ell-1}2^{-i}\\
&\quad  -2\sum_{z\in V_1\cup V_2\cup \{u_1\}}2^{-d_{G}(z,u_0)}\\
&=\left(2\sum_{z\in V_1}2^{-d_{G}(z,u_1)}-1\right)\left(\sum_{w\in V_2}2^{-d_{G}(u_1,w)} -1\right)-2^{-\ell+1}.
\end{align*}

If $\ell\ge 2$, then by Eqs.~\eqref{V1} and ~\eqref{V2},
\begin{align*}
&\quad C(G'-e)-C(G-e') \\
&\ge \left(2(b_1-2)(1-2^{-\ell-1})-1 \right)\left((b_2-1)(1-2^{-\ell-1}) -1\right) -2^{-\ell+1}\\
&\ge \left(1-2^{-2} \right)\left(1-2^{-2} \right)-2^{-1}\\
&> 0,
\end{align*}
a contradiction. It thus follows that $\ell=1$.

If there is a nontrivial block $B$ with size $b$ with  $V(B_1)\cap V(B)=\{x\}$ or $V(B_2)\cap V(B)=\{x\}$.
Note that
\[
\sum_{z\in V(T_x)}2^{-d_G(z,u_1)}\ge 2^{-1}+(b-1)2^{-2}(2-2^{-1}).
\]
Suppose first that $x\in N$. If $b_1\ge 4$, $b_2\ge 4$ or $b\ge 4$, then by Eqs.~\eqref{V1} and~\eqref{V2},
\begin{align*}
&\quad C(G'-e)-C(G-e') \\
&\ge \left(2\sum_{z\in V_1}2^{-d_{G}(z,u_1)}-1\right)\left(\sum_{w\in V_2}2^{-d_{G}(u_1,w)} -1\right)-2^{-\ell+1}\\
&= \left(2\sum_{z\in V(T_x)}2^{-d_{G}(z,u_1)}-1\right)\left(\sum_{w\in V_2}2^{-d_{G}(u_1,w)} -1\right)\\
&\quad +2\sum_{x\in V_1\setminus V(T_x)}2^{-d_{G}(z,u_1)}\left(\sum_{w\in V_2}2^{-d_{G}(u_1,w)} -1\right)-1\\
&\ge \left(1+(b-1)2^{-1}(2-2^{-1})-1\right)\left((b_2-1)(1-2^{-2}) -1\right)\\
&\quad +(b_1-3)(1-2^{-2})\left((b_2-1)(1-2^{-2}) -1\right) -1\\
&>0,
\end{align*}
a contradiction.
So $b=b_1=b_2=3$.

Note that there is no nontrivial block different from $B_1$ that has a common vertex with $B_2$, otherwise, we have
$\sum_{w\in V_2}2^{-d_G(u_1,w)}\ge 1-2^{-2}+2^{-1}+2\cdot2^{-1}(1-2^{-2})=2$ and hence \[
C(G'-e)-C(G-e')\ge (1+2-2^{-1}-1)(2-1)-1>0,
\]
a contradiction.
By Claim 4, the length of pendant path at either vertex in $V(B_2)\setminus\{u_1\}$ is one.

Let $G''=u_0u_1+\{u_1w:w\in V(B)\setminus\{x \} \}$.
Assume that $R^L(G'')=C(G''-f)$.
By similar argument as in Claims 1 and 2, $f=u_0x$.
Suppose that there is a nontrivial block different from $B$ at some vertex in $V(B)\setminus\{x\}$.
Let $T'_w$ be the component of $G-w$ containing $w$ for $w\in \{u_1\}\cup V(B)\setminus \{x\}$ and $V_B=\cup_{w\in V(B)\setminus\{x\}}V(T'_w)$.
Then
\begin{align*}
&\quad C(G''-f)-C(G''-u_0u_0')\\
&=\left(2\sum_{z\in V(T'_{u_1})}2^{-d_G(z,x)}-1 \right)\left(\sum_{w\in V_B}2^{-d_G(x,w)}-1 \right)-1\\
&\ge (1+2-2^{-1}-1)(2-1)-1\\
&>0,
\end{align*}
a contradiction.
By similar argument as in Claim 4, the length of pendant path at each vertex in $V(B)\setminus\{x \}$ is one.
Thus $G$ is a graph on $12$ vertices with $7$ cut vertices and exactly $3$ nontrivial blocks.
By direct calculation, $R^L(G)=25.625<27.375=R^L(PK_{12,7})$, a contradiction.

It follows that $x\notin N$, so $x\in V(B_2)$, which  is still impossible by similar arguments as above.
Thus,  there is no nontrivial block different from $B_1$ and $B_2$ that has a common vertex with $B_1$ or $B_2$.
By Claim 4, each pendant path at any vertex of $V(B_1)\cup V(B_2)\setminus\{u_1\}$ is of length one.  That is, $G\cong H_n(b_1, b_2)$.

Assume that $b_1\le b_2$.
If $b_1\ge 4$, then
\[
C(G'-e)-C(G-e')\ge \left(4(1-2^{-2})-1 \right)\left(3(1-2^{-2}) -1\right)-1>0,
\]
a contradiction.
So, $b_1=3$.
If $b_2\ge 6$, then
\[
C(G'-e)-C(G-e')\ge \left(2(1-2^{-2})-1 \right)\left(5(1-2^{-2}) -1\right)-1>0,
\]
a contradiction.
So $b_2\le 5$.
If $b_2=5$, then by direct calculation, we have $R^L(G')=36>35.25=R^L(G)$, a contradiction.
So we are left with two possibilities: $G\cong H_9(3,3)$ or $G\cong H_{11}(3,4)$.

Therefore, we have proved that $G\cong H_9(3,3), H_{11}(3,4)$, or there is precisely one nontrivial block of $G$.
By Lemma \ref{vertexc},
we have  $G\cong H_9(3,3)$ with $(n,k)= (9,5)$, $G\cong H_{11}(3,4)$ with $(n,k)=(11,6)$, or $G\cong PK_{n,k}$.

By direct calculation,  $R^L(H_9(3,3))=16$ and $R^L(H_{11}(3,4))=24.5$.

Let $q$ and $r$ be integers with $k=(n-k)q+r$ and $0\le r<k$. Assume that
$PK_{n,k}=K^{a_1,\dots,a_{n-k}}$ with $a_1\ge \dots\ge a_{n-k}$.
If $r=0$, then $a_1=\dots=a_{n-k}=q$, so we have
by Eq.~\eqref{MM} that
\begin{align*}
R^L(PK_{n,k})
&=2k-4+2^{2-q}+(n-k-1)2^{-q} +{{n-k-1}\choose 2}(2-2^{-q})^2.
\end{align*}
If $r\ge 1$, then $a_1=\dots=a_r=q+1$, $a_{r+1}=\dots =a_{n-k}=q$, so we have
by Eq.~\eqref{MM} that
\begin{align*}
R^L(PK_{n,k})
&=2k-4+2^{1-q}+(r-1)2^{-q-1}+(n-k-r)2^{-q}\\
&\quad +{{r-1}\choose 2}(2-2^{-q-1})^2+{{n-k-r}\choose 2}(2-2^{-q})^2\\
&\quad +(r-1)(n-k-r)(2-2^{-q-1})(2-2^{-q}).
\end{align*}
Now the proof is completed by noting that $R^L(PK_{9,5})=15.25<R^L(H_9(3,3))=16$ and $R^L(PK_{11,6})=24.5=R^L(H_{11}(3,4))$.
\end{proof}

%

\vspace{3mm}

\noindent
{\bf Data Availability} \\
All data are incorporated into the article and its online supplementary material.

\vspace{3mm}

\noindent
{\bf Acknowledgement}\\
This work was supported by  National Natural Science Foundation of China (No.~12071158).



\begin{thebibliography}{99}

\bibitem{AO1}  Aytac, A. and  Odabas,  Z.~N. (2011)
Residual closeness of wheels and related networks.
 {\em  Internat. J. Found. Comput. Sci.},  {\bf 22}, 1229--1240.

\bibitem{AO2} Aytac, A. and   Berberler, Z.~N.~O. (2017)
Robustness of regular caterpillars.
 {\em Internat. J. Found. Comput. Sci.}, {\bf 28}, 835--841.



\bibitem{AO3}
 Aytac, A. and   Odabas Berberler, Z.~N. (2018)
Network robustness and residual closeness.
 {\em RAIRO Oper. Res.},  {\bf 52},  839--847.

%


%
%
%


\bibitem{BY} Berberler, Z.~N. and Yi\v{g}it, E. (2018)
Link vulnerability in networks.
 {\em Internat. J. Found. Comput. Sci.}, {\bf 29}, 447--456.


\bibitem{Be}  Berge, C. (1958)
Sur le couplage maximum d'un graphe.
 {\em C. R. Acad. Sci. Paris}, {\bf 247}, 258--259.







\bibitem{CZ} Cheng, M. and  Zhou, B. (2022)
Residual closeness of graphs with given parameters.
 {\em J. Oper. Res. Soc. China}, Doi: 10.1007/s40305-022-00405-9.

\bibitem{Chv} Chv\'atal, V. (1973)
Tough graphs and Hamiltonian circuits.
 {\em Discrete Math.}, {\bf 5},
215--228.


\bibitem{Dan} Dangalchev, C. (2006)
Residual closeness in networks.
 {\em Phys. A}, {\bf 365}, 556--564.

\bibitem{Dan2} Dangalchev, C. (2011)
Residual closeness and generalized closeness.
 {\em Internat. J. Found. Comput. Sci.}, {\bf 22}, 1939--1948.

\bibitem{Dan3} Dangalchev, C. (2018)
Residual closeness of generalized thorn graphs.
 {\em Fund. Inform.},  {\bf 162}, 1--15.



\bibitem{Fr2}  Frank, F. and  Frisch, I.~T. (1970)
Analysis and design of survivable networks.
 {\em IEEE Trans. Commun. Tech.}, {\bf 18}, 501--519.


\bibitem{HK} Holme, P., Kim,   B.~J.,  Yoon, C.~N.,  and Han, S.~K. (2002)
 Attack vulnerability of complex networks.
  {\em Phys. Rev. E.}, {\bf 65}, 056109.




\bibitem{Ja}  Jackson, M.~O. (2008)
Social and Economic Networks,
Princeton University Press,  Princeton, New Jersey.

\bibitem{Ju}   Jung, H.~A. (1978),
On a class of posets and the corresponding comparability graphs.
  {\em J. Combin. Theory Ser. B}, {\bf 24}, 125--133.



\bibitem{OA1} Odabas, Z.~N. and Aytac, A. (2013)
Residual closeness in cycles and related networks.
 {\em Fund. Inform.}, {\bf 124}, 297--307.

\bibitem{Tu} Tutte,  W.~T. (1947) The factorization of linear graphs.
  {\em J. London Math. Soc.}, {\bf 22}, 107--111.

\bibitem{WZ} Wang, Y. and Zhou, B. (2022)
Residual closeness, matching number and chromatic number.
 {\em Comput. J.},
Doi: 10.1093/comjnl/bxac004.

\bibitem{Woo} Woodall, D.~R. (1973)
The binding number of a graph and its Anderson number.
 {\em J. Combin. Theory Ser. B}, {\bf 15},  225--255.

\bibitem{YB} Yi\v{g}it, E. and Berberler, Z.~N. (2019)
 A note on the link residual closeness of graphs under join operation.
 {\em Internat. J. Found. Comput. Sci.},  {\bf 30},  417--424.

\bibitem{YB2}Yi\v{g}it, E. and Berberler, Z.~N. (2019)
Link failure in wheel type networks.
 {\em Internat. J. Modern Phys. C}, {\bf 30},  1950072.

\bibitem{ZLG} Zhou, B.,  Li, Z., and Guo, H. (2021)
Extremal results on vertex and link residual closeness.
 {\em Internat. J. Found. Comput. Sci.}, {\bf 32}, 921--941.

\end{thebibliography}
\end{document}